\documentclass{lmcs} 


\usepackage{hyperref}
\usepackage{virginialake}
\usepackage{tikz-cd}
\usepackage{adjustbox}
\newcommand{\K}{\textnormal{\textsf{K}}}
\newcommand{\CS}{\textnormal{\textsf{CS}}}
\newcommand{\CSA}{\CS\A}
\newcommand{\GL}{\textnormal{\textsf{GL}}} 
\newcommand{\GLS}{\textnormal{\textsf{GLS}}} 
\newcommand{\CSM}{\textnormal{\textsf{CSM}}} 
\newcommand{\CSMP}{\mathsf{CSMP}}
\newcommand{\ER}{\mathsf{ER}}
\newcommand{\GLP}{\mathsf{GLP}}
\newcommand{\Sfive}{\textnormal{\textsf{S5}}}
\newcommand{\Sfour}{\textnormal{\textsf{S4}}} 
 
\newcommand{\Lo}{{\textnormal{\textsf{L}}}} 
 
\newcommand{\Kfour}{\textnormal{\textsf{K4}}}

\newcommand{\EA}{\mathsf{EA}}
\newcommand{\PA}{\mathsf{PA}}
\newcommand{\ZF}{\mathsf{ZF}}

\newcommand{\PRA}{\mathsf{PRA}}

\newcommand{\labCS}{\ell^\infty \CS}
\newcommand{\labCSX}{\ell^\infty \CS (\X)}
\newcommand{\lCS}[1]{\ell^\infty \CS (#1)}
\newcommand{\labCSXind}{\ell^{ind} \CS(\X)}
\newcommand{\lCSind}[1]{\ell^{ind} \CS(#1)}
\newcommand{\GGL}{\textnormal{G3}\GL}
\newcommand{\labGL}{\ell\textnormal{GL}}
\newcommand{\labCSXstr}{\ell^\infty \CS^{str}(\X)}

\newcommand{\vltrind}[3]{\vltr{#1}{#2}{\vlhy{\hspace{#3}}}{\vlhy{\hspace{#3}}}{\vlhy{\hspace{#3}}}}

\renewcommand{\r}{\mathsf{r}}
\newcommand{\Id}{\textnormal{Id}}
\newcommand{\botL}{\bot{\Ra}}
\newcommand{\botR}{\bot{\Ra}}
\newcommand{\toR}{{\Ra}{\to}}
\newcommand{\toL}{{\to}{\Ra}}
\newcommand{\BoxR}{{\Ra}{\Box}}
\newcommand{\triagR}{{\Ra}{\triangle}}
\newcommand{\BoxL}{{\Box}{\Ra}}
\newcommand{\triagL}{{\triangle}{\Ra}}
\newcommand{\BoxRind}{\BoxR^{ind}}
\newcommand{\triagRind}{\triagR^{ind}}
\newcommand{\wk}{\textnormal{wk}}
\newcommand{\subV}{\textnormal{sub}_{\ell}}

\newcommand{\n}[1]{\overline{#1}}
\newcommand{\msub}{[\n\x /\n \y]}
\newcommand{\wdh}{(=)}
\newcommand{\trans}{\textnormal{tr}}
\newcommand{\transitivity}{tr}
\newcommand{\transRR}{\textnormal{\transitivity}_{RR}}
\newcommand{\transRS}{\textnormal{\transitivity}_{RS}}
\newcommand{\transSR}{\textnormal{\transitivity}_{SR}}
\newcommand{\transSS}{\textnormal{\transitivity}_{SS}}

\newcommand{\Mont}{\text{M}}
\newcommand{\Pione}{\text{P}}
\newcommand{\er}{\text{ER}}

\newcommand{\Ra}{\Rightarrow}
\newcommand{\la}{\langle}
\newcommand{\ra}{\rangle}
\newcommand{\lra}{\leftrightarrow}
\newcommand{\Prov}{\textnormal{Prov}}
\newcommand{\Gdl}[1]{\ulcorner #1 \urcorner}
\newcommand{\R}{\mathcal{R}}

\newcommand{\Pw}{\mathcal{P}}

\newcommand{\Mw}{\mathcal{M}}

\newcommand{\F}{\mathfrak{F}}
\newcommand{\ph}{\varphi}
\newcommand{\A}{\mathcal{A}}
\newcommand{\PL}{\mathrm{PL}}
\newcommand{\X}{\mathsf{X}}
\newcommand{\vdashcy}{\vdash_{cyc}}
\newcommand{\La}{\mathcal{L}}
\newcommand{\Seq}{\mathfrak{S}}

\newcommand{\tree}{\mathcal{T}}
\newcommand{\LaBi}{\La^{\Box \triangle}}
\newcommand{\LaBox}{\mathcal{L}^{{\Box}}}
\newcommand{\Prop}{\mathtt{Prop}}
\newcommand{\Var}{\mathtt{Lab}}
\newcommand{\Lab}{\textnormal{Lab}}
\newcommand{\x}{\mathtt{x}}
\newcommand{\y}{\mathtt{y}}
\newcommand{\z}{\mathtt{z}}
\renewcommand{\u}{\mathtt{u}}
\renewcommand{\v}{\mathtt{v}}
\newcommand{\branch}{(\Seq_i)_{i< \omega}}
\newcommand{\trace}{(\x_i)_{i <\omega}}
\newcommand{\Smo}{{Smory{\'n}ski }}
\newcommand{\finbranch}{(\Seq_i)_{i < n}}

\begin{document}

\title[Non-Wellfounded Labelled Calculi for Bimodal Provability]{Non-Wellfounded and Labelled Sequent Calculi for Bimodal Provability Logics}

\author[J.~Becker]{Justus Becker\lmcsorcid{0000-0002-1825-0097}}

\address{University of Birmingham, United Kingdom}

\begin{abstract}
    \noindent Bimodal provability logics are modal logics where the two modalities are interpreted as two formal provability predicates of arithmetic theories.
    We present modular labelled calculi for four different bimodal provability logics, whose rules arise from the relational semantics of these logics, while the second-order property of Noetherianity of models is enforced using techniques from cyclic proof theory. 
    Completeness with respect to the semantics is established via proof search and countermodel construction.
    We further present regular and finite variants of the calculi via proof transformations.
\end{abstract}

\maketitle

\section{Introduction}

Modal provability logic is the study of formal provability predicates in a modal and usually propositional setting. 
It thus interprets the $\Box$ modality as a formal provability predicate of an arithmetic theory $T$. Doing so, one obtains what is known as the \emph{provability logic of $T$}.
The widely celebrated results from L{\"o}b~\cite{Loeb_1955} and Solovay~\cite{Solovay76} show that the provability logic of Peano Arithmetic ($\PA$) is exactly G{\"o}del-L{\"o}b provability logic ($\GL$), which can be characterised by L{\"o}b's axiom $\Box ( \Box A \to A) \to \Box A$ over the basic normal modal logic $\K$.
However, this is not a special feature of $\PA$, as it has been observed that $\GL$ is also the provability logic for an even wider class of arithmetic theories (see~\cite{Artemov2005}).

Still, there are many other modal provability logics besides $\GL$ when considering not one but a multitude of theories. 
A first generalisation, introduced independently by Artemov~\cite{Artemov80} and Visser~\cite{Visser81}, is to consider the provability interpretation of some theory $T$ under a different metatheory $U$. 
These provability logics, which are all extensions of $\GL$, have been well studied and even fully classified for all pairs of theories for which Solovay's Theorem is known to hold (see~\cite{Beklemishev_1990} for an overview).

Similarly, Carlson~\cite{Carlson86} and \Smo~\cite{Smorynski85} considered multi-modal logics where each modality is interpreted as some standard provability predicate of certain theories.
As the most interesting behaviour of these logics is the pairwise behaviour of modalities, it is of sufficient interest to only consider two modalities (call them $\Box$ and $\triangle$).
As it also minimises complexity, we consider the bimodal setting in which the metatheory, over which we interpret the provability predicates, is the intersection of both theories (i.e.\ the smaller theory if there is one).
We refer to this class of logics as \emph{bimodal provability logics}.
Note that bimodal provability logics strictly generalise the notion of provability logic assuming the theories at hand are $\Sigma_1$-sound (i.e.\ every provable $\Sigma_1$ sentence is also true in the standard model of arithmetic) by considering the bimodal theorems of the form $\triangle A$ where $A$ does not contain the $\triangle$ modality.

There are further logics, which do not fall under this definition but can still be regarded as bimodal provability logics, such as Japaridze's polymodal logic $\GLP$~\cite{Japaridze85} (or rather its bimodal fragment) or the bimodal logic of fast and slow provability $\mathsf{GLT}$~\cite{Henk16}. 
However, we will not consider these logics here.

\subsection{Contribution}
Unlike unimodal provability logics, bimodal provability logics have been much less studied and are thus also generally less understood than their unimodal counterparts.
For example, the classification problem (i.e.\ describing the class of all possible provability logics) for bimodal provability logic is still wide open, while complete for the unimodal case.
There are nine different bimodal provability logics which have been studied as of today (see~\cite{Artemov2005}). 
All of these logics are known to be decidable and have either a Kripke semantics or have a translation into such a semantics, for which they are sound and complete. 
However, there has been little investigation into the proof theory of these logics \cite{SierraMirandaStuder_2026}.
We hope to close this gap of a missing proof system in a uniform manner as it might not only lead to new perspectives on these logics but also as some of the proof-theoretic descriptions in our setting, such as for $\ER$, actually become simpler than their semantics.

In this work we employ labelled sequents, which are syntactically enriched sequents that allow us to handle modal reasoning in a principled and modular way. 
This allows us to formalise the essential properties of the relational semantics in a first-orderly manner (see e.g.~\cite{Gabbay_1996,Negri05,Vigano2000}). 
While there are simpler Gentzen-style sequent calculi for well-known modal logics such as $\K$, $\Sfour$ or $\GL$ (see e.g.~\cite{SambinValentini_1982}), their modal rules do not generally follow the sequent calculus tradition of having a left and right rule for each connective, and thus they can be seen as poorly behaved. 
Instead, they usually contain a single modal rule which analyses a multitude of formulas at once. 
This often makes these systems rather unwieldy, as proving cut-elimination can be quite difficult and non-trivial (see for example~\cite{GoreRamanayake_2012,Goreetal_2021}).
We can further observe that there are many classical normal modal logics such as $\Sfive$ or $\mathsf{S4.2}$ which do not have a cut-free Gentzen-style presentation (see~\cite{Takano_1992}).
Another downside of the pure sequent approach in general is that sequent calculi are usually not modular, i.e.\ one cannot usually add or remove rules to gain respectively a super- or sublogic.
Instead, one might need to exchange some rules for others to retain cut-free completeness.
As we are interested in different logics simultaneously, it will make our workings easier when considering calculi which can handle those in a modular manner.

However, similar to $\GL$, the relational semantics of bimodal provability logic requires its frames to be \emph{Noetherian}\footnote{This property is also called \emph{converse wellfounded}. However, as the property of wellfoundedness (or lack thereof) is reserved for proofs in this work, we use the term Noetherian for that purpose.}. 
This is a property which is not definable in first-order logic. 
To sidestep this issue, we use a well-known technique in the proof theory of modal provability logic (as well as fixpoint logics in general), by considering non-wellfounded (and cyclic) proofs. 
Non-wellfounded proofs are a generalisation of cyclic proofs as they can contain infinite branches without the need to contain repeats. They were first introduced for modal provability logic $\GL$ by Shamkanov~\cite{Shamkanov14} based on a Gentzen-style sequent approach.
Similarly, Das, van der Giessen and Marin~\cite{Das_et_al24} introduced a non-wellfounded labelled sequent calculus for $\GL$, called $\ell \GL$ (as well as a system for an intuitionistic variant of $\GL$).
Similar to our approach, we adapt the ideas behind $\ell \GL$ as a system which formalises the essential semantic notions.
Moreover, the system allows non-wellfounded proofs to contain infinite branches as long as they satisfy a global \emph{correctness condition} which semantically corresponds to Noetherianity of the frames for $\GL$. 

\subsection{Outline of the work}
The main result of this paper is outlined in Figure~\ref{fig:tour}.
Section~\ref{sec:prelim} introduces the basic definitions for bimodal provability logics, including the syntax, Hilbert axiomatisations and its Kripke-like semantics. We will also refer to the arithmetic interpretations of the logics considered. 
In Section~\ref{sect:thecalculus}, we introduce our main proof systems $\lCS{\A^r}$ for different bimodal provability logics $\CS\A$, motivated by the relational semantics, and show that they admit cyclic variants.
We prove (cut-free) completeness of $\lCS{\A^r}$ via transfinite proof search and countermodel construction in Section~\ref{sect:completeness}.
Finally, in Section~\ref{sec:inductive}, we obtain a finitary inductive system $\ell^{ind}\CS(\A^r)$ derived by a proof transformation argument from $\lCS{\A^r}$, establishing its completeness on the way, whereas soundness is proved semantically, yielding soundness of $\lCS{\A^r}$ and its cyclic variant as a corollary.

\newcommand{\tikzsize}{3.5}
\begin{figure}
    \centering
    \begin{tikzpicture}
        \node (Lprove) at (-\tikzsize,2*\tikzsize) {$\CS\A \vdash A$};
        \node (Lsemant) at (-\tikzsize,\tikzsize) {$\forall \Mw\in \mathcal{C}_{\CS\A}\; \Mw\vDash A$};
        \node (NWF) at (\tikzsize,\tikzsize) {$\ell^\infty\CS(\A^r) \vdash A$};
        \node (cyc) at (\tikzsize,2*\tikzsize) {$\ell^\infty\CS(\A^r) \vdash_{cyc}A$};
        \node (ind) at (0,0) {$\ell^{ind}\CS(\A^r)\vdash A$};

        \draw [<->] (Lprove) to node [right=0.3em, align=center] 
        {Theorem~\ref{thm:SemanticCompl} \\ \cite{Smorynski85,Carlson86,Beklemishev94,Visser95}} (Lsemant);
        \draw [] (Lprove) to node [below=0.6em] 
        {} (Lsemant);
        \draw[->] (Lsemant) to node [above=0.3em] {Theorem~\ref{thm:Completeness}} (NWF);
        \draw[<->] (NWF) to node [right=0.3em] {Proposition~\ref{prop:NWFisCyc}} (cyc);
        \draw[->] (NWF) to node [below right] {Proposition~\ref{prop:NWFtoInd}} (ind);
        \draw[->] (ind) to node [below left] {Theorem~\ref{thm:indSemanticSoundness}} (Lsemant);
    \end{tikzpicture}
    \caption{Outline of the main results, with $\A \in \{\emptyset,\{M\} , \{ M , P\} , \{ER\}\}$ and $\mathcal{C}_{\CSA}$ being the class of models of $\CSA$.}
    \label{fig:tour}
\end{figure}

\section{Bimodal Provability Logic} \label{sec:prelim}

In this section we provide the necessary background for bimodal provability logics needed for the rest of this paper. 
We begin by introducing the syntax and its arithmetic interpretation.
Then, we consider two relational semantics for the considered bimodal provability logics.
For a general reference to bimodal provability logic we refer the reader to the survey~\cite[Chapter 8]{Artemov2005}.

\subsection{Syntax and Arithmetic Semantics}

We work with a propositional language enriched with two unary modal operators $\Box$ and $\triangle$. Let $\Prop = \{ p ,q, r, ... \}$ be a countable set of propositional atoms. Well-formed formulas of the language of (classical) bimodal provability logic $\LaBi$ are constructed via the following grammar.
\[
 A :: = \bot \mid p \in \Prop \mid A \to A \mid \Box A \mid \triangle A
\]

We can further define other connectives as usual, such as $\neg A := A \to \bot$, $A \vee B := \neg A \to B$, $A \wedge B := \neg (\neg A \vee \neg B)$ and $A \lra B := (A \to B) \wedge (B \to A)$.
We might occasionally identify the two modalities as indexed modalities $\Box = \Box_0$ and $\triangle = \Box_1$; we might use indices when referring to an arbitrary modality.
We further might write $\Box_i^n$ for an iteration of $n$-many $\Box_i$ (formally $\Box_i^0 A = A$ and $\Box_i^{n+1} A = \Box_i \Box_i^n A$).
We call $\LaBox$ the $\triangle$-free part of the bimodal language $\LaBi$.
We define the \emph{uniform substitution} of $p$ with $B$ in the formula $A$, called $A[B/p]$, by replacing every occurrence of $p$ in $A$ by the formula $B$.

From now on, let $T$ and $U$ be recursively enumerable (r.e.), $\Sigma_1$-sound theories which can interpret Elementary Arithmetic ($\EA$), also referred to as I$\Delta_0 {+} \exp$ (see e.g.~\cite{Artemov2005}). 
For simplicity, we may assume that $T$ and $U$ are both formulated in some first-order language $\La^1$ with equality. 
For both theories, we obtain standard provability predicates $\Prov_T (\Gdl{\cdot})$ and $\Prov_U (\Gdl{\cdot})$ via their elementary presentation due to Craig's trick~\cite{Feferman_1960}.
Such theories include most well-known mathematical theories such as Peano Arithmetic ($\PA$), Primitive Recursive Arithmetic ($\PRA$), Zermelo-Fraenkel set theory ($\ZF$), as well as, for example, $\PA + Con(\PA)$.

We interpret our bimodal syntax in terms of provability of $T$ and $U$ via the map $\cdot^\ast : \LaBi \to \La^1$ defined as follows (where $\La^1$ only contains closed sentences): 
\begin{itemize}
    \item[$\bullet$] $\bot^\ast = (0=1)$, 
    \item[$\bullet$] $p^\ast \in \La^1$, 
    \item[$\bullet$] $(A \to B)^\ast = A^\ast \to B^\ast$, 
    \item[$\bullet$] $(\Box A)^\ast = \Prov_T (\Gdl{A^\ast})$, and 
    \item[$\bullet$] $(\triangle A)^\ast = \Prov_U (\Gdl{A^\ast})$. 
\end{itemize}

The bimodal provability logic of the pair $(T,U)$ is then defined as the set of all bimodal formulas which are provable under any interpretation in the common metatheory $T \cap U$, formally:
\[
\PL(T,U) := \{ A \in \LaBi \mid \textnormal{ for all interpretations }\cdot^\ast, T \cap U \vdash A^\ast \}.
\]
Whereas the (unimodal) provability logic of a theory $T$ over a metatheory $U$ is defined as all unimodal formulas which are provable under any interpretation in $U$, formally:
\[
\textnormal{PL}_U(T) := \{ A \in \LaBox \mid \textnormal{for all interpretations }\cdot^\ast, U \vdash A^\ast \}.
\]

Recall the axiomatisation of G{\"o}del-L{\"o}b provability logic ($\GL$).

\begin{figure}
    \[
    \vliinf{mp}{}{B}{A}{A \to B} \qquad \qquad \vlinf{n_\Box}{}{\Box A}{A} \qquad \qquad \vlinf{n_\triangle}{}{\triangle A}{A} \qquad \vlinf{sub}{}{A[B/p]}{A}
    \]
    \caption{Hilbert rules of normal bimodal logics}
    \label{fig:HilbertRule}
\end{figure}

\begin{defi}
    The logic $\GL$ is axiomatised over classical propositional logic by the following schemata
    \begin{itemize} \setlength{\itemindent}{2em}
    \item[$(\textnormal{K}_\Box)$]
            $\Box (A \to B) \to (\Box A \to \Box B)$
    \item[$(\textnormal{L}_\Box)$] 
          $\Box (\Box A \to A ) \to \Box A$
    \end{itemize}
    and closed under the rules of modus ponens $(mp)$ and necessitation $(n_\Box)$ (see Figure~\ref{fig:HilbertRule}).
\end{defi}

Note that one can also axiomatise $\GL$ by including the 4-axiom of modal logic ($\Box A \to \Box \Box A$). 
While this axiom can be derived already in $\GL$, it makes clear that $\GL$ is an extension of the modal logic $\Kfour$.
It has been shown by L{\"o}b~\cite{Loeb_1955} that the provability interpretation of modal logic contains $\GL$ whereas Solovay~\cite{Solovay76} showed that $\GL$ faithfully interprets the provability behaviour of $\PA$. 
Moreover, we have the following general \emph{Solovay Theorem}.

\begin{thmC}[\cite{Artemov2005}]
    $\GL = \PL_T(T)$. \qed
\end{thmC}

Recall for this theorem that $T$ is some first-order, r.e.\ and $\Sigma_1$-sound theory.\footnote{For this theorem to hold as stated, $T$ could still be $\Sigma_1$-unsound as long as it has a characteristic of $\infty$ (see~\cite{Artemov2005} for details).}
Thus, this theorem implies G{\"o}del's Second Incompleteness Theorem for any such theory $T$ as $\GL \vdash \neg \Box \bot \to \neg \Box \neg \Box \bot$ via contraposition of L{\"o}b's axiom.

In the bimodal setting however, the story of a canonical provability logic is not that straightforward, as the resulting logic depends heavily on the relationship between the two theories that are chosen. 
Of course, considering the logic yielded by PL$(T,T)$ is just $\GL$ together with the axiom schema $\Box A \lra \triangle A$, which is rather uninteresting.
Nevertheless, there is a smallest bimodal logic which is contained in any bimodal provability logic. 
This logic is called Carlson-\Smo ($\CS$), named after Carlson and \Smo due to their foundational work in multimodal provability logic~\cite{Carlson86,Smorynski85}.

\begin{defi}
    The logic $\CS$ is axiomatised over $\GL$ by the following axiom schemata 
   \begin{itemize} \setlength{\itemindent}{2em}
    \item[$(\textnormal{K}_\triangle)$] 
           $\triangle (A \to B) \to (\triangle A \to \triangle B)$
    \item[$(\textnormal{L}_\triangle)$] 
            $\triangle (\triangle A \to A ) \to \triangle A$ 
    \item[$(\textnormal{4}_{\triangle\Box})$]
          $\Box A \to \triangle \Box A$ 
    \item[$(\textnormal{4}_{\Box\triangle})$]
        $\triangle A \to \Box \triangle A$ 
   \end{itemize}
  and the rules of modus ponens $(mp)$, $\Box$-necessitation $(n_\Box)$ and $\triangle$-necessitation $(n_\triangle)$ (see Figure~\ref{fig:HilbertRule}).
\end{defi}

Observe that $\CS$ contains the rules and axioms of $\GL$ for both modalities $\Box$ and $\triangle$; further,
any unimodal fragment of $\CS$ (either $\Box$ or $\triangle$ free) gives back the logic $\GL$ (up to renaming of modalities).
The additional transitivity-like axioms 4$_{\triangle \Box}$ and 4$_{\Box\triangle}$ come from the fact that the provability predicate is $\Sigma_1$ while $\EA$ (and thus all its extensions) is provably $\Sigma_1$-complete (i.e.\ $\EA \vdash \ph \to \Prov_{\EA}(\Gdl{\ph})$ for $\ph \in \Sigma_1$).
We refer to $\CS$ as the \emph{basic bimodal provability logic} due to $\CS \subseteq \PL(T,U)$ for any theories $T$ and $U$ (see~\cite{Artemov2005}). 

While there are theories $T$ and $U$ such that $\CS = \PL (T,U)$~\cite{Smorynski85}, these theories do not arise naturally. 
Let us consider extensions of $\CS$, some of which correspond to natural examples of theories (see Figure~\ref{fig:inclusion-of-logics}). 
We denote extensions of a logic $\Lo$ with a set of axioms $\A$ closed under ($mp$), ($n_\Box$), ($n_\triangle$) and uniform substitution ($sub$) as $\Lo\A$.\footnote{As we defined $\CS$ over axiom \emph{schemata}, the logic $\CS$ is already closed under uniform substitution.}

Firstly, we consider the case where one theory extends another. 
Without loss of generality, we assume $U$ to be the larger theory, thus making $\triangle$ the stronger modality (note that this convention is sometimes reversed, such as in~\cite{Visser95}). 
The provability logic of such a pair of theories validates the following monotonicity principle $M$.\footnote{As suggested by Visser~\cite{Visser95}, the $M$ might also stand for Montagna.}
\[
(M) \;\; \Box A \to \triangle A.
\]

Thus, the logic $\CSM :=\CS\{M\}$ is included in PL$(T,U)$ if $T \subseteq U$~\cite{Smorynski85}. 

Let now $U$ be an extension of $T$ by $\Pi_1$-axioms. 
Then, the following principle will always be validated in PL$(T,U)$~\cite{Beklemishev94}:
\[
(P) \;\; 
\triangle A \to \Box (\triangle \bot \vee A)
\]

Informally, this axiom arises from the fact that the weaker theory $T$ can verify a theorem of its $\Pi_1$-axiomatised extension $U$ whenever $T$ assumes the consistency of $U$.
This is a consequence of there being a truth predicate for $\Pi_1$ sentences (see~\cite{Beklemishev94} for details).
We call the resulting logic $\CSMP := \CSM\{P\}$. 

Finally, let us consider the logic $\ER := \CS \{ ER\}$ (for \emph{essentially reflexive} extension) with 
\[
(ER) \;\; \triangle (\Box A \to A).
\]

The logic $\ER$ corresponds to pairs of theories $(T,U)$ for which $U$ proves the full reflection principle of $T$, i.e.~$U \vdash \Prov_T(\Gdl{\ph}) \to \ph$ for all $\ph \in \La^1$~\cite{Carlson86}. As a natural example, $\ER$ coincides with the provability interpretations of the following pairs: $\ER =$ PL$(\PA,\ZF)=$ PL$(I\Sigma_n , \PA)$ (see e.g.~\cite{Smorynski85}). 
$\ER$ proves the principle $M$ (but not $P$); therefore, $\ER$ properly extends $\CSM$ while it is incomparable with $\CSMP$.
Furthermore, it is known that $\ER$ is a maximal bimodal provability logic, meaning that there are no pairs of consistent theories $(T,U)$ such that PL$(T,U)$ properly extends $\ER$ (note that there are consistent extensions of $\ER$ in general such as $\ER\{\triangle\bot\}$ or $\ER\{P\}$).

\renewcommand{\tikzsize}{2}
\begin{figure}
    \centering
    \begin{tikzpicture}
        \node (CS) at (0,0){$\CS$};
        \node (CSM) at (0,\tikzsize){$\CSM$};
        \node (CSMP) at (-\tikzsize, 2*\tikzsize){$\CSMP$};
        \node (ER) at (\tikzsize, 2*\tikzsize){$\ER$};

        \draw[-] (CS) to node  {}(CSM);
        \draw[-] (CSM) to node {} (CSMP);
        \draw[-] (CSM) to node {} (ER);
    \end{tikzpicture}
    \caption{Hasse diagram of the inclusions of bimodal provability logics considered in this work}
    \label{fig:inclusion-of-logics}
\end{figure}

Observe that, given a bimodal provability logic $\PL(T,U)$, we can extract the provability logic $\PL_U(T)$ by $(\PL(T,U))^\Box = \{ A \in \LaBox \mid \PL(T,U) \vdash \triangle A\}$, which in the literature is also called \emph{type}, \emph{$\triangle$-type} or \emph{1-type}.
Note that the converse does not hold, i.e.\ a bimodal provability logic is not determined by its 1-type.
All bimodal logics that we have presented, except $\ER$, have $\GL$ as their type~\cite{Artemov2005}.
We briefly introduce the logic which is the type of $\ER$.

\begin{defi}
    The logic $\GLS$ (for G{\"o}del-L{\"o}b-Solovay) is axiomatised over $\GL$ by the axiom schema
    \[
    (T) \;\; \Box A \to A
    \]
    and the rule of modus ponens.
\end{defi}
Note that here, we do not close the logic under necessitation.
$\GLS$ is also called the \emph{true provability logic} of $\PA$ as it captures all sentences which are true (in the standard model) about the provability of $\PA$. 
This is also known as \emph{Solovay's Second Completeness Theorem}~\cite{Solovay76}.
It is further known due to the classification theorem for provability logics~\cite{Beklemishev_1990} that all logics of the shape $\PL_U (T)$ where $U$ and $T$ are sound lie between $\GL$ and $\GLS$. 

\begin{propC}[\cite{Carlson86}] \label{prop:type-ER-GLS}
    $\GLS = \ER^\Box$. \qed
\end{propC}

Note that the few bimodal provability logics introduced here are far from exhaustive.
All of the logics we intentionally omitted here are either Kripke-incomplete (i.e.~there is no sound and complete class of Kripke frames for these logics) or have essentially only models of infinite size (see~\cite{Artemov2005} for details). 
While this might not necessarily be an obstacle for designing proof systems, the known semantic completeness results allow us to obtain easier cut-admissibility results for our systems without the need to perform internal cut-elimination.\footnote{The logic $\ER$ is actually Kripke-incomplete, however there is a semantics for general frames, which has been developed by Visser~\cite{Visser95}, that we can utilise.} 

\subsection{Kripke Semantics}

Recall the Kripke semantics of $\GL$.

\begin{defi}
    A $\GL$-frame $\F = \la W , \prec \ra$ is a pair consisting of a set of points (also called worlds) $W \neq \emptyset$ and a relation ${\prec} \subseteq W \times W$ on points, such that
    \begin{enumerate}
        \item $\prec$ is transitive, i.e.\ if $x \prec y \prec z$ then $x \prec z$
        \item $\prec$ is Noetherian (also called converse wellfounded or terminating), i.e.\ there is no infinite ascending chain $x_1 \prec x_2 \prec x_3 \prec \dots$
    \end{enumerate}

    A $\GL$-model $\Mw = \la W , \prec , V \ra$ is a $\GL$-frame $\la W , \prec \ra$ together with a valuation function $V : \Prop \to \Pw (W)$.
\end{defi}

Note that Noetherianity of $\GL$-frames implies that they must be irreflexive and cycle-free.
A modal formula $A \in \LaBox$ is locally evaluated at a point $x$ in a model $\Mw$ (of $\GL$) inductively over its structure.
\begin{itemize} 
    \item[$\bullet$] $\Mw , x \nvDash \bot$
    \item[$\bullet$] $\Mw , x \vDash p$ iff $x \in V(p)$
    \item[$\bullet$] $\Mw , x \vDash A \to B$ iff $\Mw , x \vDash B$ or $\Mw , x \nvDash A$
    \item[$\bullet$] $\Mw , x \vDash \Box A$ iff for all $y$ such that $x\prec y$ we have $\Mw , y \vDash A$
\end{itemize}

We say that a model $\Mw = \la W, \prec , V \ra$ validates a formula $A$, formally written $\Mw \vDash A$, iff $\Mw , x \vDash A$ for every $x \in W$. We also define for a frame $\F$ and a point $x$ that $\F, x \vDash A$ iff $(\F , V) , x \vDash A$ for all valuations $V$. 
Similarly, a frame validates a formula, written as $\F \vDash A$, iff $(\F, V) \vDash A$ for all valuations (equivalently, $\F , x \vDash A$ for all points $x$).
We might also phrase the following theorem in terms of $\GL$-models.

\begin{thmC}[\cite{Segerberg_1971}] \label{thm:GLcompl}
    $\GL \vdash A$ iff $\F \vDash A$ for all $\GL$-frames $\F$. \qed
\end{thmC}

\subsection{Carlson and \Smo Semantics}

The (non-arithmetic) semantics for bimodal provability logics behaves similarly to the relational semantics for $\GL$. 
There are two approaches to extend Kripke semantics to serve as semantics for bimodal provability logics which can be found in the literature. 
Both approaches are essentially equivalent, as we will show. 
However, they are often conflated without being explicit on how these kinds of models relate.
We call the first kind of structures \emph{Carlson frames}, which are $\GL$-frames enriched with two subsets of points, one for each modality. 

\begin{defiC}[\cite{Carlson86}]
    A \emph{Carlson frame} $\F$ is a tuple $\la W, \prec , M_0 , M_1 \ra$ such that $\la W, \prec \ra$ is a $\mathsf{GL}$-frame with two subsets $M_0 , M_1 \subseteq W$. 
    A \emph{Carlson model} is a tuple $\la W, \prec , M_0 , M_1 , V\ra$ consisting of a Carlson frame $\la W, \prec , M_0 , M_1 \ra$ and a valuation function $V : \Prop \to \Pw(W)$.
\end{defiC}

A priori, there are no conditions on the subsets $M_0$ and $M_1$. 
A modal formula $A \in \LaBi$ is evaluated on a Carlson model $\Mw$ just like before with only the modal cases differing.
\begin{itemize} 
    \item[$\bullet$] $\Mw , x \vDash \Box A$ iff for all $y$ such that $x\prec y$ and $y \in M_0$ we have $\Mw , y \vDash A$
    \item[$\bullet$] $\Mw , x \vDash \triangle A$ iff for all $y$ such that $x\prec y$ and $y \in M_1$ we have $\Mw , y \vDash A$
\end{itemize}

A different but very similar approach is to include an additional relation, which is standard for considering multimodal logics, giving us what we call \emph{\Smo frames}.

\begin{defiC}[\cite{Smorynski85}] \label{def:Smorynski}
    A \emph{\Smo frame} $\F$ is a tuple $\la W, \prec_0, \prec_1 \ra$ such that $W \neq \emptyset$ is a set of points and ${\prec_0, \prec_1} \subseteq W \times W$ are relations on $W$ such that 
    \begin{enumerate}
        \item $\prec_0 $ and $ \prec_1 $ are transitive,
        \item ${\prec_0} \circ {\prec_1} \subseteq {\prec_1}$,
        \item ${\prec_1} \circ {\prec_0} \subseteq{ \prec_0} $, and
        \item ${\prec_0} \cup {\prec_1} $ is Noetherian.
    \end{enumerate}
    A \emph{Smory\'nski model} is a tuple $\la W, \prec_0, \prec_1, V\ra$ consisting of a Smory\'nski frame $\la W, \prec_0 , \prec_1 \ra$ and a valuation function $V : \Prop \to \Pw(W)$.
\end{defiC}

A formula is evaluated in a Smory\'nski model as before, only with the modal cases differing.

\begin{itemize}
    \item[$\bullet$] $\Mw , x \vDash \Box A$ iff for all $y$ such that $x\prec_0 y$ we have $\Mw , y \vDash A$
    \item[$\bullet$] $\Mw , x \vDash \triangle A$ iff for all $y$ such that $x\prec_1 y$ we have $\Mw , y \vDash A$
\end{itemize}

We define $\vDash$ for both Carlson and Smory\'nski frames and models, as we did for $\GL$-frames and models.
While Carlson models arise more naturally when considering the logic $\CS$ itself, Smory\'nski models are more natural to handle stronger logics, since they simplify the conditions on the two relations.
Still, it is a well-known fact that both semantics can be seen as notational variants of each other.
However, as there is no explicit proof of their strong connection, we include a proof of how one can translate between Carlson frames and equivalent Smory\'nski frames.
Given a class of frames $\mathcal{C}$, we call $Log(\mathcal{C}) := \{ A \in \LaBi \mid \forall \F \in \mathcal{C}. \F \vDash A \}$ the logic validated by this class.

\begin{prop} \label{prop:Smo-Carlson}
    The following are equivalent for a given logic $\Lo \subseteq \LaBi$:
    \begin{enumerate}
        \item There is a class of Carlson frames $\mathcal{C}^C$ with $Log(\mathcal{C}^C)=\Lo$.
        \item There is a class of Smory\'nski frames $\mathcal{C}^S$ with $Log(\mathcal{C}^S)=\Lo$.
    \end{enumerate}
\end{prop}
\begin{proof}
    We show this by transforming Carlson frames into elementarily equivalent Smory\'nski frames and vice versa.

    $(1)\ \Longrightarrow (2)$: For any given Carlson frame $\la W^C , \prec , M_0 , M_1 \ra = \F^C \in \mathcal{C}^C$ we define the \Smo transformation $S(\cdot)$ by $S(\F^C) := \la W^C , \prec_0 , \prec_1\ra$ with ${\prec_i} : = \{ (x,y) \in W^C \times W^C \mid x \prec y$ and $y \in M_i \}$. 
    It is easily shown by induction on the structure of the formula $A$ that for any valuation $V : \Prop \to \Pw(W^C)$: $\la \F^C , V\ra, x \vDash A$ iff $\la S(\F^C), V\ra , x \vDash A$ for all points $x \in W^C$.
    Thus, we have $ \F^C \vDash A$ iff $S(\F^C) \vDash A$.
    Therefore, taking the $S$-image $S[\mathcal{C}^C]$ of a class of Carlson frames $\mathcal{C}^C$ gives us a class of frames satisfying the same formulas as $\mathcal{C}^C$.

    $(2)\ \Longrightarrow (1)$: For any given \Smo frame $\la W^S , \prec_0 , \prec_1 \ra =  \F^S \in \mathcal{C}^S$ we define a Carlson transformation $C(\cdot)$ by $C(\F^S) := \la W^S , \prec , M_0 , M_1 \ra$ with ${\prec} := {\prec_0} \cup {\prec_1}$ and $M_i = \{ x \in W^S \mid \exists y (y \prec_i x) \}$.
    It is easily shown by induction on the structure of the formula $A$ that for any valuation $V: \Prop \to \Pw(W^S)$: $\la \F^S , V\ra , x \vDash A$ iff $\la C(\F^S), V\ra , x \vDash A$ for all points $x \in W^S$.
    Thus, we have $\F^S \vDash A$ iff $ C(\F^S) \vDash A$.
    Therefore, taking the $C$-image $C[\mathcal{C}^S]$ of a class of \Smo frames $\mathcal{C}^S$ gives us a class of frames satisfying the same formulas as $\mathcal{C}^S$.
\end{proof}

This shows that it generally does not matter which kind of semantics we use.
Note that the previous proof also allows us to easily translate between certain kinds of frame conditions of Carlson and Smory\'nski semantics.

\begin{figure}
    \centering
    \begin{tikzpicture}
        \node (1) at (0,0){$x$};
        \node (2) at (0,\tikzsize){$y$};
        \node (3) at (0,2*\tikzsize){$z$};
        \node (a) at (0,-0.5*\tikzsize){$\F_1$};

        \node (4) at (2*\tikzsize,0){$x$};
        \node (5) at (2*\tikzsize,\tikzsize){$y$};
        \node (6) at (2*\tikzsize,2*\tikzsize){$z$};
        \node (b) at (2*\tikzsize,-0.5*\tikzsize){$\F_2$};

        \node[label={[xshift=-5pt]right:$\in M_1$}] (7) at (4*\tikzsize,0){$x$};
        \node[label={[xshift=5pt]left:$M_0 \ni$}] (8) at (4*\tikzsize,\tikzsize){$y$};
        \node[label={[xshift=-5pt]right:$\in M_1$}] (9) at (4*\tikzsize,2*\tikzsize){$z$};
        \node (c) at (4*\tikzsize,-0.5*\tikzsize){$\F_3$};

        \draw[->] (1) to node [left]{$\prec_0$}(2);
        \draw[->] (2) to node [left]{$\prec_1$}(3);
        \draw[->, bend right=25] (1) to node [right]{$\prec_0$} (3);

        \draw[->] (4) to node [left]{$\prec_0$}(5);
        \draw[->] (5) to node [left]{$\prec_1$}(6);
        \draw[->, bend right=25] (4) to node [right]{$\prec_1$} (6);

        \draw[->] (7) to node [left]{$\prec$}(8);
        \draw[->] (8) to node [left]{$\prec$}(9);
        \draw[->, bend right=25] (7) to node [right]{$\prec$} (9);
    \end{tikzpicture}
    \caption{The structure $\F_1$ satisfies all conditions of a \Smo frame except condition (2) ${\prec_0} \circ {\prec_1} \subseteq {\prec_1}$, $\F_2$ is a \Smo frame, and $\F_3$ is a Carlson frame that is equivalent to $\F_2$.}
    \label{fig:three-frames}
\end{figure}

\begin{exa} \label{ex:Frames}
    Consider the structures given in Figure~\ref{fig:three-frames}.
    The underlying relations $\prec$ and ${\prec_0} \cup {\prec_1}$ are the same in each structure. 
    However, $\F_1$ is not a proper \Smo frame as it violates the condition ${\prec_0} \circ {\prec_1} \subseteq {\prec_1}$.
    This makes the formula schema $\triangle A \to \Box \triangle A$ falsifiable in $\F_1$.
    For example, in the model $\Mw = \la \F_1 , V \ra$ with the valuation $V(p)= \emptyset$, we have $\Mw , x \vDash \triangle p$, as well as $\Mw , z \nvDash p$, therefore $\Mw , y \nvDash \triangle p$ and thus $\Mw , x \nvDash \Box \triangle p$; and therefore, $\Mw, x \nvDash \triangle p \to \Box \triangle p$.
    The \Smo frame $\F_2$ is equivalent to the Carlson frame $\F_3$.
    Moreover, we have $\F_2 = S(\F_3)$ with $S(\cdot)$ being the \Smo transformation from the proof of Proposition~\ref{prop:Smo-Carlson}.
    Note that $x \in M_1$ in $\F_3$ cannot be represented in the frame $\F_2$, thus $\F_3 \neq C(\F_2)$; however, this has no impact on the evaluation of formulas.
\end{exa}

For convenience, we take Smory\'nski frames and models as our primitive structures from now on, which we will also just call \emph{frames} and \emph{models} respectively. 
We prefer to work with birelational semantics as it requires fewer symbols to work with and allows for even greater modularity, e.g.\ $\triangle A \to \Box \triangle  A$ is always validated in a (transitive) Carlson frame but can be made invalid in a transitive birelational frame (see Example~\ref{ex:Frames}).
From now on, we consider $M_0$, $M_1$ and $\prec$ as auxiliary symbols whenever convenient, defined by $M_i := \{x \mid \exists y . y \prec_i x \}$ and ${\prec} := {\prec_0} \cup {\prec_1}$.
We will also use auxiliary symbols $x \preceq_i y$ (and $x \preceq y$) to mean $x \prec_i y$ or $x = y$ (respectively $x \prec y$ or $x = y$).

Most of the semantic completeness results that we present here can be seen in terms of correspondence between modal axioms and frame conditions.
For example, the axiom $M = \Box A \to \triangle A$ corresponds to the frame condition of ${\prec_1} \subseteq {\prec_0}$ (equivalently $M_1 \subseteq M_0$).
For logics extending $\CSM$, one also finds in the literature (see~\cite{Beklemishev94,Beklemishev_1996}) frames of the form $\la W , \prec , M \ra$ which one can read as equivalent Carlson frames $\la W , \prec , M_0 , M_1 \ra$ with $W = M_0$ and $M = M_1$.
We might also observe modal correspondence for the $\CS$ axioms of $4_{\triangle \Box}$ and $4_{\Box\triangle}$ with the respective conditions ${\prec_1} \circ {\prec_0} \subseteq {\prec_0}$ and ${\prec_0} \circ {\prec_1} \subseteq {\prec_1}$.

For the semantics of $\CSMP$, we consider the following frame property. 
A Carlson frame $\la W , \prec , M_0,M_1 \ra$ is \emph{downwards closed} if the set $M_1$ is downwards closed with respect to $\prec$ (i.e.\ if $y \in M_1$ and $x \prec y$ then $x \in M_1$). 
Equivalently we call a \Smo frame $\la W , \prec_0 , \prec_1 \ra$ downwards closed if, whenever ${x \prec_0 y \prec_1 z}$ then $x \prec_1 y$.
We call a model downwards closed if its underlying frame is downwards closed.

Even though $ER$ is also an axiom which has an easy first-order correspondence, the viewpoint of correspondence breaks down for the logic $\ER$.
The axiom $\triangle (\Box A\to A)$ only holds on frames which contain some reflexive point (or $\prec_1$ is empty which would imply $\ER \vdash \triangle \bot$ which cannot be the case); but note that frames of $\CS$ must already be irreflexive. 
Thus, $\ER$ is frame-incomplete.
One way to alleviated this is by considering \emph{general frames}, which naturally occur when obtaining relational semantics from algebraic semantics (see \cite{BlackburnRijkeVenema01}). Thus, any normal modal logic is always complete for some class of general frames.
One might also view general frames as defining a specific class of models by restricting not only the relations but also the possible valuations. 
To define the class of general $\ER$-frames, described by Visser~\cite{Visser95}, let us define further frame conditions.

\begin{defi} \label{def:limit-point}
    A point $x \in W$ in a frame $\F = \la W , \prec_0 , \prec_1  \ra $ is called a \emph{limit point} in $\F$ if the following holds (see also Figure~\ref{fig:LimitPoint}): 
    
    (i) there is some $y \in W$ such that $x\prec_0 y$, and 
    
    (ii) for every $y,z \in W$, if $x\prec_0 y$ and $x \prec_0 z$, then there is some $u \in W$ such that $x \prec_0 u$, $u\prec_0 y$ and $u\prec_0 z$.
\end{defi}

\begin{figure}
    \[\begin{tikzcd}
	y && y && z \\
	&&& u \\
	x &&& x
	\arrow[dashed, from=2-4, to=1-3]
	\arrow[dashed, from=2-4, to=1-5]
	\arrow[dashed, from=3-1, to=1-1]
	\arrow[from=3-4, to=1-3]
	\arrow[from=3-4, to=1-5]
	\arrow[dashed, from=3-4, to=2-4]
\end{tikzcd}\]
    \caption{Frame conditions of a limit point $x$. 
    Solid arrows represent stipulated pairs, and dashed arrows represent inferred pairs of points belonging to the relation $\prec_0$.}
    \label{fig:LimitPoint}
\end{figure}

Unlike frames, general frames restrict the possible valuations their models can take. We do this by defining the \emph{inclusive upset} $x\!\uparrow_{\preceq} : = \{ y \in W \mid x\preceq y\}$ and the \emph{strict downset} $x\!\downarrow_{\prec} := \{ y \in W \mid y\prec x \}$ of a point $x$ in a frame. In~\cite{Visser95}, these sets are called basic clopens forming a clopen Hausdorff topology which form the admissible sets of the general semantics. This topological interpretation is also in part related to the topological semantics of provability logics~\cite{Esakia_1981,Beklemishev10}.

\begin{defi} \label{def:ER-semantics}
    A \emph{general $\ER$-frame} is a tuple $\la W , \prec_0 ,\prec_1 , \tau \ra $ such that $ {\prec_1} \subseteq {\prec_0}$, all points in $M_1 :=\{ x \mid \exists y. y \prec_1 x \}$ are limit points in $\la W, \prec_0 ,\prec_1 \ra$, and $\tau \subseteq \Pw(W)$ is the set of \emph{admissible valuations}, which is defined as follows: $\tau$ contains the sets $x\!\uparrow_{\preceq}$ and $ x\!\downarrow_{\prec}$ for every $x \in W$ and is closed under finite unions and complements. 

    An \emph{$\ER$-model} is a structure $\la W , \prec_0 , \prec_1 , V \ra $ where $V(p) \in \tau$ for every $p \in \Prop$ and $\la W , \prec_0 , \prec_1 , \tau \ra $ is a general $\ER$-frame.
\end{defi}

\renewcommand{\tikzsize}{1.5}
\begin{figure}
    \centering
    \begin{tikzpicture}
        \node (omega1) at (0,0) {$\omega+1$};
        \node (omega) at (\tikzsize,0) {$\omega$};
        \node (empty) at (2*\tikzsize,0){};
        \node (2) at (3*\tikzsize,0){2};
        \node (1) at (4*\tikzsize,0){1};
        \node (0) at (5*\tikzsize,0){0};

        \draw[->, bend right=40](omega1) to node [below=0.5em]{$\prec_1$} (omega);
        \draw[->, bend left=40](omega1) to node [above=0.5em]{$\prec_0$} (omega);
        \draw[loosely dotted](1.5*\tikzsize,0) -- (2*\tikzsize,0);
        \draw[->](empty) to node [above=0.5em]{$\prec_0$}(2);
        \draw[->](2) to node [above=0.5em]{$\prec_0$}(1);
        \draw[->](1) to node [above=0.5em]{$\prec_0$}(0);
    \end{tikzpicture}
    \caption{The smallest general frame of the logic $\ER$ satisfying ${\prec_1} \neq \emptyset$ (transitive arrows are intentionally omitted for readability)}
    \label{fig:example-ER-frame}
\end{figure}

\begin{exa}
    The smallest $\ER$-frame $\la W, \prec_0, \prec_1 ,\tau\ra$, where $\prec_1$ is non-empty, can be defined by the ordinal $\omega +2 = : W$ with ${\prec_0} : = {\supset}$ being the proper co-inclusion, and ${\prec_1} := \{ (\omega+1 , \omega)\}$ (Figure~\ref{fig:example-ER-frame}). The admissible valuations $\tau := \{ V \mid V \subseteq_{fin}W$ or $\omega \in V \subseteq W\}$ of the frame consist of finite subsets of $W$ and co-finite subsets of $W$ which contain $\omega$.
    It is an easy exercise to check that $\tau$ is indeed the set of admissible valuations and $\omega$ is a limit point.
\end{exa}

Crucially, the property of limit points together with the admissible valuations for an $\ER$-model give these points what one might call ``pseudo-reflexivity''.

\begin{propC}[\cite{Visser95}] \label{prop:limitPoint}
    Let $\F = \la W , \prec_0 , \prec_1 , \tau\ra $ be a general $\ER$-frame and $x \in W$. We have $\F, x \vDash \Box A \to A$ iff $x$ is a limit point. \qed
\end{propC}

Altogether, we have the following semantic completeness theorems for our bimodal provability logics.

\begin{thmC}[\cite{Smorynski85,Carlson86,Beklemishev94,Visser95}]\label{thm:SemanticCompl}
    For any formula $A \in \LaBi$ we have the following: 
    \begin{enumerate}
        \item $\mathsf{CS} \vdash A$ iff $\F \vDash A$ for all (Smory{\'n}ski) frames $\F$.
        \item $\CSM \vdash A$ iff $\F \vDash A$ for all frames $\F$ with ${\prec_1} \subseteq {\prec_0}$.
        \item $\CSMP \vdash A$ iff $\F \vDash A$ for all downwards closed frames $\F$ with ${\prec_1} \subseteq {\prec_0}$.
        \item $\ER \vdash A$ iff $\F \vDash A$ for all general $\ER$-frames $\F$. \qed
    \end{enumerate}
\end{thmC}

Note that this theorem trivially also holds for the classes of models which are induced by the respective classes of frames.
We will also throughout this paper refer to models and frames as \emph{belonging to a logic $\Lo$} (or \emph{$\Lo$-models}) whenever they satisfy the corresponding conditions.

\section{Labelled Sequent Calculi} \label{sect:thecalculus}

We begin this section by covering the basics about labelled proof theory for the unimodal provability logic $\GL$.
We then introduce our basic non-wellfounded and labelled calculus $\labCS$ which we can extend by a set of rules $\X$ to $\labCSX$.
This gives us a calculus $\lCS{\A^r}$ for all of the logics we have introduced in the previous section by identifying the set of extension axioms $\A$ with their respective rules $\A^r$.
We show that the calculi are complete with respect to their respective logics via transfinite proof search and countermodel construction. 
In another step, we show that our non-wellfounded system admits regular proofs. 

\subsection[Labelled calculi for GL]{Labelled Calculi for \texorpdfstring{$\GL$}{}}

The use of labelled systems can be traced back to the work of Kanger~\cite{Kanger57}, 
who used them to introduce the first cut-free system for the modal logic $\mathsf{S5}$. However, only much later were labelled systems applied widely to non-classical logics~\cite{Gabbay_1996,Negri05,Simpson1994,Vigano2000} including multimodal logic~\cite{GargGenoveseNegri_2012}.

For the language of (simple) labelled sequents we use a countable set of \textit{labels} $\Var=\{\x,\y,\z,...\}$ as well as a relation symbol $R$. 
Note that we distinguish labels $\x, \y, \z$ from points in models $x,y,z$ by utilising different fonts to help separate more easily what happens on the level of semantics and sequents.
A \textit{labelled formula} is then formed by $\x:A$ with $\x \in \Var$ and $A \in \LaBi$, and a \emph{relational atom} is of the form $\x R\y$.
\emph{Labelled sequents} are then written as $\R , \Gamma \Ra \Theta$, where $\R$ is a (finite) set of relational atoms, while $\Gamma$ and $\Theta$ are (finite) sets of labelled formulas.
As usual, we write $\Gamma , \Gamma'$ for the union of two sets $\Gamma$ and $\Gamma'$, while we also write $\Gamma , \x : A$ and $\R , \x R \y$ for $\Gamma , \{\x : A\}$ and $\R ,\{ \x R \y\}$ respectively.
We use $\Seq$ as a metavariable for labelled sequents, and we refer to the set of labels occurring in it as $\Lab(\Seq)$.  
We say that a sequent $\R , \Gamma \Ra \Theta$ is \emph{contained} in another sequent $\R' , \Gamma' \Ra \Theta'$, written as $(\R , \Gamma \Ra \Theta) \subseteq (\R' , \Gamma' \Ra \Theta')$ if $\R \subseteq \R'$, $\Gamma \subseteq \Gamma'$ and $\Theta \subseteq \Theta'$. 

\begin{figure}[t]
\[
\vlinf{\Id}{}{\R , \Gamma , \x : p \Ra \Theta, \x : p }{} \qquad 
\vlinf{\botR}{}{\R , \Gamma , \x : \bot \Ra \Theta}{}
\]
\[
\vlinf{\toR}{}{\R , \Gamma \Ra \Theta, \x : A \to B }{\R , \Gamma, \x: A \Ra \Theta , \x : B }\qquad
\vliinf{\toL}{}{\R , \Gamma, \x : A \to B \Ra \Theta}{\R , \Gamma \Ra \Theta , \x : A}{\R , \Gamma, \x: B \Ra \Theta}
\]
\[
\vlinf{\BoxR_\GL}{\y!}{\R, \Gamma \Ra \Theta , \x : \Box A }{\R, \x R \y  , \Gamma , \y : \Box A \Ra \Theta , \y : A } \qquad
\vlinf{\BoxL}{}{\R, \x R \y , \x : \Box A , \Gamma \Ra \Theta}{\R, \y : A, \Gamma \Ra \Theta}
\]
\[
\vlinf{\trans}{}{\R, \x R \y , \y R \z , \Gamma \Ra \Theta}{\R , \x  R \z, \Gamma \Ra \Theta}
\]
    \caption{The rules of $\GGL$}
    \label{fig:RulesG3GL}
\end{figure}

The rules of the first labelled sequent system for the logic $\GL$, called $\GGL$ and due to Negri~\cite{Negri05}, are shown in Figure~\ref{fig:RulesG3GL}.\footnote{Note that our system is a slight adaptation of Negri's as it uses sets as cedents instead of multisets.
But, as weakening and contraction are admissible in $\GGL$, this constitutes no essential difference.} 
We call the sets of labelled formulas and relational atoms $\R$, $ \Gamma$ and $\Theta$ which are left unchanged by a rule instance the \emph{context} of the rule instance.
The labelled formulas and the relational atoms which occur outside the context in the lower or upper sequent of a rule instance are called respectively \emph{principal} and \emph{auxiliary}. 
We call the auxiliary formula $\y : \Box A$ in $\BoxR_\GL$ the \emph{diagonal formula}.
The lower sequents of rules with no premisses ($\Id$ and $\botL$) are called \emph{axiomatic sequents}.
Note that we can always choose the principal formula(s) to occur in the context as we work with cedents as sets which do not keep track of multiplicities of elements.
We call such rule instances \emph{maximal}.

While the operational rules are standard for a sequent calculus, the $\BoxL$ rule can be understood via the standard translation of modal formulas, i.e.\ reading $\x : \Box A$ as $\forall \y ( \x R\y \to \y :A)$.
This only partially applies to the rule $\BoxR_\GL$, due to the additional diagonal formula $\y : \Box A$. 
This can be motivated due to it expressing the fact that the relation $R$ must be Noetherian (see Lemma \ref{lem:FalsityI} and Proposition \ref{prop:NWFtoInd}).
However, we can also see its resemblance in L{\"o}b's rule, which is equivalent to L{\"o}b's axiom over $\Kfour$ (see~\cite[page 59]{Boolos94}):
\[
\vlinf{\text{L{\"o}b}}{}{A}{\Box A \to A}
\]
Lastly, we also have a rule $\trans$ capturing the condition that the relation has to be transitive. 

We say that a labelled sequent $\Seq$ is provable in $\GGL$, written $\GGL \vdash \Seq$, if there is a proof of this sequent.
A formula $A$ is provable, $\GGL \vdash A$, whenever we have $\GGL \vdash \Ra \x :A$ for some $\x \in \Var$. 

\begin{thmC}[\cite{Negri05}]
    $\GL \vdash A$ iff $\GGL \vdash A$. \qed
\end{thmC}

The second proof-theoretic tools we use are \emph{non-wellfounded} and \emph{cyclic proofs}. 
Cyclic proofs were first introduced for the modal $\mu$-calculus~\cite{Niwinski96} to handle fixpoints.
While wellfounded proofs are generally inductively defined trees, non-wellfounded proofs are defined necessarily as co-inductive trees, allowing for branches of infinite length. 
The first non-wellfounded and cyclic proof systems for the logic $\GL$ were introduced by Shamkanov~\cite{Shamkanov14}. 
The first non-wellfounded and labelled system for provability logic, called $\labGL$, was introduced by Das, van der Giessen and Marin~\cite{Das_et_al24}.
Both these non-wellfounded systems are based on different sequent calculi for the logic $\Kfour$, neither of which contain a diagonal formula in their $\Box$-right rules.
Moreover, $\labGL$ has the same rules as $\GGL$, except that the $\BoxR_\GL$ rule is replaced by the following rule.
\[
\vlinf{\BoxR}{\y!}{\R, \Gamma \Ra \Theta,\x : \Box A }{\R, \x R \y , \Gamma \Ra \Theta,\y : A }
\]
Furthermore, $\ell \GL$ allows for infinite branches in proofs under a progressing trace condition:

\begin{defiC}[\cite{Das_et_al24}]\label{def:traceGL}
     A \emph{pre-proof} is a possibly infinite tree generated by rules of $\labGL$.

    A \emph{trace} along a branch $(\Seq_i)_{i<\omega}$ (with $\Seq_i = \R_i , \Gamma_i \Ra \Theta_i$) of a pre-proof is a sequence of labels $(\x_i)_{i < \omega}$ s.t.~for all $i<\omega$, one of the following holds
\begin{enumerate}
\item $\x_{i+1} = \x_i$; or
\item ${\x_i R \x_{i+1} \in \R_i}$.
\end{enumerate}    
    We call $\x_{i+1}$ a \emph{progress point} if $\x_i R \x_{i+1}$.
    
    A trace is \emph{progressing} if it contains infinitely many progress points. A branch $(\Seq_i)_{i<\omega}$ in a pre-proof is \emph{progressing} if it has a progressing trace.
    
    A \emph{proof} in the system $\labGL$ is then given by a pre-proof where each branch is either finite or progressing. 
\end{defiC}

Note that every sequent in a pre-proof must be a conclusion of a rule instance.
Trees of sequents generated by rules which might have \emph{open leaves}, i.e.\ 
sequents which are not conclusions of rule instances, are referred to as \emph{derivations}.
We can understand the progress condition as saying that in every infinite branch in a proof we can find an infinite chain of labels which are connected via relational atoms.

\begin{thmC}[\cite{Das_et_al24}]
    $\GL \vdash A$ iff $\labGL \vdash A$. \qed
\end{thmC}

\subsection{Labelled Calculi for Bimodal Provability Logic}

Like the previously introduced calculus $\ell\GL$, the calculi we present are based on the relational semantics introduced in the previous section. 

For the language of our calculi $\labCSX$, we extend the notion of labelled sequent by also allowing for another relation symbol $S$, which can occur as $\x S \y$ in the set of relational atoms $\R$ with $\x ,\y \in \Var$.
Whenever convenient, we refer to relation symbols in general by an indexed relation symbol $R_i$ with $R_1 := S$ and $R_0 := R$, similar to how we might write the two modalities as $\Box_0$ and $\Box_1$.

We define the \emph{(label) substitution} of $\x$ for $\y$ over a sequent $\Seq$ by substituting every occurrence of $\y$ in $\Seq$ with the label $\x$; we write the sequent obtained from this as $\Seq[\x/\y]$. 
We might also write $\R[\x/\y],\Gamma[\x/\y]\Ra \Theta[\x/\y]$ for the sequent obtained by substituting $\x$ for $\y$ in $\R , \Gamma \Ra \Theta$.

Like before, the rules of $\labCS$ (Figure~\ref{fig:RulesLabCS}) are motivated by the semantic conditions for connectives.
For example, we have transitivity rules for each combination of relation symbols just like the transitivity conditions for \Smo models (Definition~\ref{def:Smorynski}).
This connection is made explicit in Definition~\ref{def:SequInterpret}.
Note that we work with explicit structural rules of weakening, label substitution as well as a repeat rule denoted by $\wdh$.
While the repeat rule is trivially admissible in wellfounded proofs, its occurrences (possibly infinitely many in one branch) may impact provability in the non-wellfounded setting.
For these structural rules we will show that they can be eliminated from non-wellfounded proofs (Lemma~\ref{lem:structural-elim}), although they will be necessary for obtaining cyclic proofs.
As before, we have contraction built into the system by having defined cedents as sets.
For symmetry reasons, one might also include a right rule for $\bot$ as $\vlinf{{\Ra}\bot}{}{\R , \Gamma \Ra \Theta , \x :\bot}{\R , \Gamma \Ra \Theta}$, which we can derive using weakening.

\begin{figure}[t]
\textit{Axiomatic Rules}
\[
\vlinf{\Id}{}{\R , \Gamma , \x : p \Ra \Theta,\x : p }{} \qquad 
\vlinf{\botL}{}{\R , \Gamma , \x : \bot \Ra \Theta}{}
\]

\medskip
\textit{Operational Rules}
\[
\vlinf{\toR}{}{\R , \Gamma \Ra \Theta,\x : A \to B }{\R , \Gamma, \x: A \Ra \Theta ,\x : B }\qquad
\vliinf{\toL}{}{\R , \Gamma, \x : A \to B \Ra \Theta}{\R , \Gamma \Ra \Theta,\x : A}{\R , \Gamma, \x: B \Ra \Theta}
\]
\[
\vlinf{\BoxR}{\y!}{\R, \Gamma \Ra \Theta,\x : \Box A }{\R, \x R \y , \Gamma \Ra \Theta,\y : A } \qquad
\vlinf{\BoxL}{}{\R, \x R \y , \Gamma,\x : \Box A  \Ra \Theta}{\R, \Gamma ,\y : A \Ra \Theta}
\]
\[
\vlinf{\triagR}{\y!}{\R, \Gamma \Ra \Theta ,\x : \triangle A }{\R, \x S \y , \Gamma \Ra \Theta,\y : A } \qquad
\vlinf{\triagL}{}{\R, \x S \y , \Gamma,\x : \triangle A  \Ra \Theta}{\R, \Gamma ,\y : A \Ra \Theta}
\]

\medskip
\textit{Transitivity Rules}
\[
\vlinf{\transRR}{}{\R, \x R \y , \y R \z , \Gamma \Ra \Theta}{\R, \x  R \z, \Gamma \Ra \Theta} \qquad
\vlinf{\transSS}{}{\R, \x S \y , \y S \z , \Gamma \Ra \Theta}{\R, \x  S \z, \Gamma \Ra \Theta}
\]
\[
\vlinf{\transRS}{}{\R, \x R \y , \y S \z , \Gamma \Ra \Theta}{\R, \x  S \z, \Gamma \Ra \Theta} \qquad
\vlinf{\transSR}{}{\R, \x S \y , \y R \z , \Gamma \Ra \Theta}{\R, \x  R \z, \Gamma \Ra \Theta}
\]

\medskip
\textit{Structural Rules}
\[
\vlinf{\wk}{}{\R , \R' , \Gamma , \Gamma' \Ra \Theta , \Theta'}{\R , \Gamma  \Ra \Theta} \qquad
\vlinf{\subV}{}{\R{[\x/\y]},\Gamma{[\x/\y]} \Ra \Theta{[\x/\y]}}{\R , \Gamma \Ra \Theta} \qquad
\vlinf{\wdh}{}{\R, \Gamma \Ra \Theta}{\R , \Gamma \Ra \Theta}
\]
    \caption{The rules of $\labCS$}
    \label{fig:RulesLabCS}
\end{figure}

Like for $\labGL$, a crucial feature of this calculus is that it allows for infinite branches which do not end in an axiomatic sequent. 

Informally speaking, we want to allow infinite branches as long as the models we can extract from those branches are not Noetherian.
We will give a formal argument corresponding to this idea in the proof of Theorem~\ref{thm:Completeness} using the following definition of proofs in $\labCS$.

\begin{defi}[Pre-Proof, Trace, Progress, Proof] \label{def:trace}
     A \emph{pre-proof} is a possibly infinite tree generated by rules of $\labCS$.

    A \emph{trace} along a branch $(\Seq_i)_{i<\omega}$ (with $\Seq_i = \R_i , \Gamma_i \Ra \Theta_i$) of a pre-proof is a sequence of labels $(\x_i)_{i < \omega}$ s.t.~for all $i<\omega$ one of the following holds
    \begin{enumerate}
        \item $\Seq_i$ concludes an instance of $\subV$ and one of the following holds
        \begin{enumerate}
            \item $\Seq_i = \Seq_{i+1}[\x/\y]$ and $\x \neq \x_{i+1} = \x_i$, or
            \item $\Seq_i = \Seq_{i+1}[\x_i /\x_{i+1}]$;
        \end{enumerate}
        \item $\Seq_i$ does not conclude an instance of $\subV$ and one of the following holds
        \begin{enumerate}
        \item $\x_{i+1} = \x_i$,
        \item ${\x_i R \x_{i+1} \in \R_i}$, or
        \item $\x_i S \x_{i+1} \in \R_i$. 
        \end{enumerate}    
    \end{enumerate}
    We call $\x_{i+1}$ a \emph{progress point} if $\x_i R \x_{i+1}$ or $\x_i S \x_{i+1}$.
    
    An infinite trace is \emph{progressing} if it contains infinitely many progress points. A branch $(\Seq_i)_{i<\omega}$ in a pre-proof is \emph{progressing} if it has a progressing trace.
    
    A \emph{proof} in the system $\labCS$ is then given by a pre-proof where each branch is either finite or progressing. 
\end{defi}

Note that this condition subsumes the one given in Definition \ref{def:traceGL}, as condition (1) can only occur if we have a rule for label substitution. 
Thus, $\labCS$ fully subsumes the system $\labGL$ from Das, van der Giessen and Marin~\cite{Das_et_al24}.

\begin{exa} \label{ex:countermodels}
    Consider the $\labCS$-proof given in Figure \ref{fig:CSproof}.
    The proved formula is obtainable in the Hilbert system of $\CS$ from an instance of L{\"o}b's axiom and the axiom $\Box p \to \triangle \Box p$.
    In the leftmost branch of the tree, we iteratively apply the rules $\BoxL$, $\toL$, $\BoxR$ and $\transRR$, but there is no axiomatic rule which we could ever apply to close the branch.
    In a wellfounded proof system, this would allow us to read out a countermodel in the limit of this branch. 
    However, as this branch contains a progressing trace starting with the distinct labels $\x , \y , \z $ and $ \u$, we do find an infinite $R$-chain. 
    Thus, the countermodel we extract in the limit of this infinite branch is not Noetherian and so not a model of $\CS$. 
\end{exa}

\begin{figure}[t]
    \begin{adjustbox}{max width=\textwidth}
    \centering
    $
    \vlderivation{
    \vlin{\toR}{}{\Ra \x : \Box ( \Box p \to p) \to \triangle \Box p}{
    \vlin{\triagR}{}{\x : \Box ( \Box p \to p) \Ra \x : \triangle \Box p}{
    \vlin{\BoxR}{}{\x S \y , \x : \Box ( \Box p \to p) \Ra \y : \Box p}{
    \vlin{\transSR}{}{\x S \y , \y R \z, \x : \Box ( \Box p \to p) \Ra \z : p}{
    \vlin{\BoxL}{}{\x S \y , \y R \z, \x R \z, \x : \Box ( \Box p \to p) \Ra \z : p}{
    \vliin{\toL}{}{\x S \y , \y R \z, \x R \z, \x : \Box ( \Box p \to p) , \z : \Box p \to p \Ra \z : p}{
       \vlin{\BoxR}{}{\R, \x : \Box ( \Box p \to p) \Ra \z : p,  \z : \Box p}{
       \vlin{\transRR}{}{\R, \z R \u,  \x : \Box ( \Box p \to p) \Ra \z : p,  \u : p}{
       \vlin{}{}{\R, \z R \u , \x R \u , \x : \Box ( \Box p \to p) \Ra \z : p,  \u : p}{\vlhy{\vdots}}
       }
       } 
    }{
        \vlin{\Id}{}{\R, \x : \Box ( \Box p \to p) , \z : p \Ra \z : p}{\vlhy{}}
    }
    }
    }
    }
    }
    }
    }
    $
    \end{adjustbox}
    \caption{A proof in $\labCS$ with $\R = \x S\y ,\y R\z , \x R\z $. The continuation, denoted by vertical dots, is defined as iterative application of $\BoxL$, $\toL$, $\BoxR$ and $\transRR$ in the leftmost branch while closing all other branches by applications of $\Id$.}
    \label{fig:CSproof}
\end{figure}

\subsection[Extensions of lCS]{Extensions of \texorpdfstring{$\labCS$}{}}

\begin{figure}[t]
    \[
    \vlinf{\Mont}{}{\R , \x S\y, \Gamma \Ra \Theta}{\R ,  \x R\y, \Gamma \Ra \Theta} \qquad
    \vlinf{\Pione}{}{\R , \x R \y, \y S \z , \Gamma \Ra \Theta}{\R , \x S \y, \Gamma \Ra \Theta} \qquad
    \vlinf{\er}{}{\R , \x S\y , \Gamma , \x : \Box A  \Ra \Theta}{\R , \Gamma , \y: A  \Ra \Theta }
    \]
    \caption{Extension rules for $\labCSX$ with $\X \subseteq \{\Mont,\Pione,\er\}$}
    \label{fig:extRules1}
\end{figure}

An advantage of labelled systems is their ability to be adapted to different logics in an easy and modular way. To utilise this, let us call $\labCSX$ the calculus $\labCS$ together with some additional set of rules $\X\subseteq \{\Mont,\Pione,\er\}$ (see Figure~\ref{fig:extRules1}).

The extension rules can be motivated by the correspondence between frame conditions and modal axioms. 
The most straightforward way to implement these is by adding rules manipulating the relational context $\R$ directly, such as $\Mont$ and $\Pione$. 
We call such rules \emph{relational}. 
Other relational rules include the transitivity rules ($\transRR, \transSS, \transRS, \transSR$) that were already present in $\labCS$.
One might also implement extensions as \emph{logical rules} where, instead of editing the relational structure, modal formulas are moved directly along the relations.
However, for convenience, we will stick with relational rules whenever possible.
The only rule which we give as a logical rule is $\er$.
This is because its relational variant would add a reflexive relational atom $\y R \y$, which would allow us to create progressing branches in pre-proofs as soon as we can apply the rule $\triagR$. 
Another approach for defining a system for $\ER$ would be to include the conditions of limit points (see Definition~\ref{def:limit-point}) as relational rules, however this ignores the conditions on the admissible valuations of $\ER$-models, for which we would need further rules or restrictions on our proofs. 
Thus, the rule $\er$ is a simplification compared to formalising \emph{all} of the semantics of $\ER$.

Note that not all combinations of rules will correspond to a logic from the previous sections.
For example, we do not consider the logic $\CS\{P , ER\}$ as it is not arithmetically complete.
Thus, in light of semantic correspondence, we will identify sets of axioms $\A \in \{ \emptyset , \{M\} , \{M,P\} , \{ER\}\}$ with their respective extension rules $\A^\r \subseteq \{\Mont , \Pione , \er \}$ as follows.
If $\A = \emptyset$ then $\A^\r = \emptyset$, if $\A = \{M\}$ then $\A^\r = \{\Mont\}$, if $\A= \{M , P\}$ then $\A^\r = \{\Mont, \Pione\}$, and if $\A= \{ER\}$ then $\A^\r = \{ \Mont , \er\}$. 
We might then also refer to these extended systems by $\lCS{\Mont} = \ell^\infty \CSM$, $\lCS{\Mont, \Pione} = \ell^\infty \CSMP$ and $\lCS{\Mont, \er} = \ell^\infty \ER$.
Nevertheless, we will state properties of $\labCS$ and its extensions without reference to the set of axioms $\A$ and instead consider any extension $\labCSX$ whenever possible, for convenience and clarity.

\begin{exa}
    Consider the non-wellfounded proof in the calculus $\ell^\infty \CSMP$, represented in Figure~\ref{fig:CSMPproof}.
    The formula that is proved is obtainable in the Hilbert system of $\CSMP$ via the axioms of L$_\triangle$ and M.
    Note that we use an auxiliary rule for negation which is admissible under the definition of negation ($\neg A := A \to \bot$):
    \[
    \vlinf{\neg{\Ra}}{}{\R, \Gamma, \x : \neg A \Ra \Theta}{\R, \Gamma\Ra \Theta, \x : A} \qquad := \qquad
    \vlderivation{
    \vliin{\toL}{}{\R, \Gamma,\x : A \to \bot \Ra \Theta}{\vlhy{\R, \Gamma \Ra \Theta, \x : A}}{\vlin{\bot{\Ra}}{}{\R, \Gamma, \x : \bot \Ra \Theta}{\vlhy{}}}
    }
    \]
    Further note that, like in the previous example, we obtain an infinite leftmost branch which has a progressing trace beginning with the distinct labels $\x, \y $ and $ \u$.
\end{exa}

\begin{figure}[t]
    \begin{adjustbox}{max width=\textwidth}
    \centering
    $
    \vlderivation{
    \vlin{\toR}{}{\Ra \x : \triangle (\triangle p \to p) \to \Box ( \neg \triangle \bot \to p)}{
    \vlin{\BoxR}{}{\x : \triangle (\triangle p \to p) \Ra  \x : \Box ( \neg \triangle \bot \to p)}{
    \vlin{\toR}{}{\x R \y , \x : \triangle (\triangle p \to p) \Ra  \y : \neg \triangle \bot \to p}{
    \vlin{\neg{\Ra}}{}{\x R \y , \x : \triangle (\triangle p \to p), \y : \neg \triangle \bot \Ra  \y : p}{
    \vlin{\triagR}{}{\x R \y ,  \x : \triangle (\triangle p \to p) \Ra  \y : p, \y : \triangle \bot}{
    \vlin{\Pione}{}{\x R \y , \y S \z , \x : \triangle (\triangle p \to p) \Ra  \y : p, \z : \bot}{
    \vlin{\triagL}{}{\x R \y , \y S \z , \x S \y , \x : \triangle (\triangle p \to p) \Ra  \y : p, \z : \bot}{
    \vliin{\toL}{}{\x R \y , \y S \z , \x S \y , \x : \triangle (\triangle p \to p), \y : \triangle p \to p \Ra  \y : p, \z : \bot}{
        \vlin{\triagR}{}{\R, \x : \triangle (\triangle p \to p) \Ra  \y : p, \z : \bot, \y : \triangle p}{
        \vlin{\transSS}{}{\R , \y S \u , \x : \triangle (\triangle p \to p) \Ra  \y : p, \z : \bot, \u : p}{
        \vlin{\triagL}{}{\R , \y S \u , \x S \u, \x : \triangle (\triangle p \to p) \Ra  \y : p, \z : \bot, \u : p}{
        \vlin{\toL}{}{\R , \y S \u , \x S \u, \x : \triangle (\triangle p \to p) , \u : \triangle p \to p \Ra  \y : p, \z : \bot, \u : p}{\vlhy{\vdots}}
        }
        }
        }
    }{
        \vlin{\Id}{}{\R , \x : \triangle (\triangle p \to p), \y : p \Ra  \y : p, \z : \bot}{\vlhy{}}
    }
    }
    }
    }
    }
    }
    }
    }
    }
    $
    
    \end{adjustbox}
    \caption{A proof in $\ell^\infty \CSMP$ with $\R =\x R \y , \y S \z , \x S \y $. The continuation, denoted by vertical dots, is defined as iterative application of $\toL$, $\triagR$, $\transSS$ and $\triagL$ in the leftmost branch while closing all other branches by applications of $\Id$.}
    \label{fig:CSMPproof}
\end{figure}

One might even consider further and more general extension rules such as bimodal Scott-Lemmon axioms~\cite{Lemmon1977} which we can usually implement in labelled proof theory (e.g.\ in a calculus where $\Box$ and $\triangle$ only satisfy the theorems of basic modal logic $\K$) using relational rules (see~\cite{Negri05}):
\[
\neg \Box_i^k \neg \Box_j^l A \to \Box_g^m \neg \Box_h^n \neg A \quad \text{ with } i , j , g , h \in \{0,1\} \text{ and } k,l,m,n \in \omega
\]

However, adding these extensions has the potential to make the logic $\CS$ either Kripke-incomplete (for example $\triangle A \to \triangle \neg \Box \neg A$ which is an axiom of the bimodal provability logic $\mathsf{CSMPR}$, see~\cite{Beklemishev94}) or inconsistent (for example $\Box A \to A$).
Thus, showing completeness might rely on syntactic cut-elimination, which is why we leave their study for future work. 

\subsection{Some Properties of the Systems}

In this subsection we prove some general results about $\labCSX$ for some set $\X \subseteq \{ \Mont, \Pione , \er\}$ of extension rules, which we fix throughout this subsection.
These results will be of use in the following sections. 
Moreover, we will show that our system $\labCSX$ admits the eliminability of structural rules $\{\wk , \subV, \wdh\}$.
We call the system $\labCSX$ without structural rules $\labCSXstr := \labCSX \setminus \{\wk , \subV, \wdh\}$.

\begin{prop}\label{prop:GenID}
    Any sequent of the form $\R, \Gamma , \x: A \Ra \Theta , \x :A $ is provable in $\labCS^{str}$ (and thus all its extensions) via finite proofs.
\end{prop}
\begin{proof}
    By induction on the structure of $A$. 

    If $A = p$, the sequent is an instance of $\Id$. If $A= \bot$, the sequent is an instance of $\botL$. Both of these constitute finite proofs.

    Let $A = B \to C$. By the inductive hypothesis there are two finite proofs (note that the context can be arbitrary):
    \[
    \vlderivation{\vltrind{\pi_1}{\R, \Gamma ,\x: B \Ra \Theta , \x : C , \x :B }{4mm}}
     \qquad
     \vlderivation{\vltrind{\pi_2}{ \R , \Gamma , \x: C , \x : B  \Ra \Theta , \x :C}{4mm}}
    \]
    By combining both proofs, we have the following finite proof.
    \[
    \vlderivation{
    \vlin{\toR}{}{\R , \Gamma , \x: B \to C \Ra  \Theta , \x : B \to C}{
    \vliin{\toL}{}{\R , \Gamma , \x: B \to C, \x : B \Ra  \Theta , \x : C}{
        \vltrind{\pi_1}{\R , \Gamma , \x : B \Ra \Theta ,  \x : C, \x: B}{4mm}
    }{
        \vltrind{\pi_2}{\R , \Gamma , \x: C, \x : B \Ra  \Theta , \x : C}{4mm}
    }
    }
    }
    \]

    Let $A = \Box_i B$. By the inductive hypothesis, there is a finite proof 
    \[
    \vlderivation{\vltrind{\pi}{\R, \x R_i \y , \Gamma , \y:  B \Ra \Theta , \y : B}{4mm}}
    \]
    from which we can obtain a finite proof as required:
    \[
    \vlderivation{
    \vlin{{\Ra}\Box_i}{}{\R , \Gamma , \x: \Box_i B \Ra \Theta , \x : \Box_i B}{
    \vlin{\Box_i{\Ra}}{}{\R, \x R_i \y , \Gamma , \x: \Box_i B \Ra \Theta , \y : B}{
    \vltrind{\pi}{\R, \x R_i \y , \Gamma , \y:  B \Ra \Theta , \y : B}{4mm}
    }
    }
    }\qedhere
    \]
\end{proof}

We call a rule \emph{admissible} if we can transform a proof of its premiss into a proof of its conclusion; we call it \emph{eliminable} if we can transform a proof with arbitrary occurrences of that rule into a proof without any occurrences.
This differentiation is important for non-wellfounded proofs since it might happen that we need to eliminate an infinite number of rule applications, which is not guaranteed to be possible.
Thus, eliminability implies admissibility, whereas the converse only holds necessarily for finitary proof systems.
In the following, we show that the structural rules in our system can be eliminated from proofs.

\begin{lem} \label{lem:sub-elim}
    The rule $\subV$ is eliminable in $\labCSX$.
\end{lem}
\begin{proof}
    Let $\pi$ be a proof of a sequent $\Seq$ in $\labCSX$. 
    We define a stepwise corecursive transformation of $\pi$ to obtain a proof of $\Seq$ without the rule $\subV$ present.
    To do this, we strengthen the statement to the eliminability of multi-substitution
    \[
    \vlinf{m\subV}{}{\R[\n\x /\n \y] , \Gamma[\n\x /\n \y] \Ra \Theta[\n\x /\n \y]}{\R , \Gamma \Ra \Theta}
    \]
    where $\n \x = (\x_1 , ... , \x_n)$ and $\n \y = (\y_1 , ... , \y_n)$ are $n$-tuples of labels such that $\y_j \neq \y_i$ for all $i,j \leq n$ and $i \neq j$, and $[\n\x /\n \y] := [\x_1 /\y_1]\dots[\x_n /\y_n]$.
    Note that $m\subV$ generalises label substitution; thus, we may replace any instances of $\subV$ in $\pi$ by instances of $m\subV$.
    We also must temporarily generalise the definition of traces (Definition~\ref{def:trace} case (1)) so that we consider instances of $m\subV$ instead of $\subV$.

    We now show that the rule $m\subV$ is eliminable from a given $\labCSX$-proof $\pi$ by a corecursive transformation of $\pi$ into a proof which is free of any instances of $m\subV$.
    We define each corecursive step by considering some lowest instance of $m\subV$ and permute that instance upwards, depending on the shape of the preceding rule (i.e.\ the rule occurring above the $m\subV$ instance).
    Throughout all cases we define the label $\x'$ to be $\x_k$ if $\x = \y_k$ for some $k \leq n$, otherwise $\x' := \x$.
    
    $\Id$: 
    \[
    \vlderivation{
    \vlin{m\subV}{}{\R[\n\x /\n \y] , \Gamma[\n\x /\n \y], \x' : p \Ra \Theta[\n\x /\n \y] , \x' :p }{
    \vlin{\Id}{}{\R , \Gamma , \x :p \Ra \Theta ,\x : p }{\vlhy{}} 
    }
    }\quad \rightsquigarrow \quad 
    \vlinf{\Id}{}{\R[\n\x /\n \y] , \Gamma[\n\x /\n \y], \x' : p \Ra \Theta[\n\x /\n \y] ,\x' :p}{}
    \]

    $\botL$: Similar.

    $\toR$: 

    \begin{adjustbox}{width=\textwidth}
        $\vlderivation{
    \vlin{m\subV}{}{\R\msub , \Gamma \msub \Ra \Theta\msub , \x' : A \to B}{
    \vlin{\toR}{}{\R , \Gamma  \Ra \Theta , \x : A \to B }{
    \vltrind{\pi}{\R , \Gamma , \x : A \Ra \Theta , \x : B }{4mm}
    }
    }
    }\quad \rightsquigarrow \quad 
    \vlderivation{
    \vlin{\toR}{}{\R\msub , \Gamma \msub \Ra \Theta\msub , \x' : A \to B}{
    \vlin{m\subV}{}{\R\msub , \Gamma \msub , \x' : A \Ra \Theta\msub , \x' : B }{
    \vltrind{\pi}{\R , \Gamma , \x : A \Ra \Theta , \x : B }{4mm}
    }
    }
    }$
    \end{adjustbox}

    $\toL$, $\Box_i{\Ra}$, tr$_{R_iR_j}$, $\Mont$, $\Pione$, $\er$: Similar.

    $\BoxR_i$: 
    Let $\y'$ be a fresh label in $\R\msub, \Gamma\msub \Ra \Theta\msub , \x' : \Box_i A$, which is possibly different from $\y$ as we might have $\y = \x_k$ for some $k \leq n$.

    \[
    \vlderivation{
    \vlin{m\subV}{}{\R\msub, \Gamma\msub \Ra \Theta\msub , \x' : \Box_i A}{
    \vlin{\BoxR_i}{}{\R, \Gamma \Ra \Theta , \x : \Box_i A}{
    \vltrind{\pi}{\R, \x R_i \y, \Gamma \Ra \Theta , \y : A}{4mm}
    }
    }
    } \rightsquigarrow
    \vlderivation{
    \vlin{\BoxR_i}{}{\R\msub, \Gamma\msub \Ra \Theta\msub , \x' : \Box_i A}{
    \vlin{m\subV}{}{\R\msub, \x' R_i \y' ,\Gamma\msub \Ra \Theta\msub , \y' : A}{
    \vltrind{\pi}{\R, \x R_i \y ,\Gamma \Ra \Theta , \y : A}{4mm}
    }
    }
    }
    \] 
    
    Note that the context of the lower sequent is equal to $\R\msub[\y'/\y], \Gamma\msub[\y'/\y] \Ra \Theta\msub[\y'/\y]$ as $\y'$ was fresh in the premiss of $\BoxR_i$. 
    Thus, the instance of multi-substitution which also substitutes $\y$ for $\y'$ is well-defined.

    $\wk$:
    \[
    \vlderivation{
    \vlin{m\subV}{}{\R\msub, \R'\msub , \Gamma\msub , \Gamma'\msub \Ra \Theta\msub , \Theta'\msub}{
    \vlin{\wk}{}{\R , \R' , \Gamma , \Gamma' \Ra \Theta , \Theta'}{
    \vltrind{\pi}{\R , \Gamma  \Ra \Theta}{4mm}
    }
    }
    } \]\[\quad \rightsquigarrow \quad 
    \vlderivation{
    \vlin{\wk}{}{\R\msub, \R'\msub , \Gamma\msub , \Gamma'\msub \Ra \Theta\msub , \Theta'\msub}{
    \vlin{m\subV}{}{\R\msub , \Gamma\msub \Ra \Theta\msub}{
    \vltrind{\pi}{\R , \Gamma  \Ra \Theta}{4mm}
    }
    }
    }
    \]

    $\subV$:
    \[
    \vlderivation{
    \vlin{m\subV}{}{\R[\x/\y]\msub , \Gamma [\x/\y]\msub \Ra \Theta[\x/\y]\msub}{
    \vlin{\subV}{}{\R[\x/\y] , \Gamma [\x/\y] \Ra \Theta[\x/\y]}{
    \vltrind{\pi}{\R , \Gamma \Ra \Theta}{4mm}
    }
    }
    }\]\[\quad \rightsquigarrow \quad 
    \vlderivation{
    \vlin{\wdh}{}{\R[\x/\y]\msub , \Gamma [\x/\y]\msub \Ra \Theta[\x/\y]\msub}{
    \vlin{m\subV}{}{\R[\x/\y]\msub , \Gamma [\x/\y]\msub \Ra \Theta[\x/\y]\msub}{
    \vltrind{\pi}{\R, \Gamma \Ra \Theta}{4mm}
    }
    }
    }
    \]
    
    To check that the new instance of $m\subV$ is correct, the only case to consider is if $\y = \y_k$ for some $k \leq n$.
    Then, as $\y_k$ cannot occur in the intermediate sequent $\R[\x/\y] , \Gamma [\x/\y] \Ra \Theta[\x/\y]$, we can remove the substitution $[\x_k/ \y_k]$ from $\msub$.
    Thus, we can use the proper multi-substitution $[\x_1 /\y_1]\dots[\x_{k-1} /\y_{k-1}][\x_{k+1} /\y_{k+1}]\dots[\x_n /\y_n][\x / \y]$ instead of $\msub[\x/\y]$ to gain the same conclusion sequent.

    $\wdh$: Trivial.

    It is left to show that this transformation yields a correct proof. 
    Observe that each transformation step is either productive or terminating; thus, the transformation yields a pre-proof.
    
    To show that the final pre-proof is indeed a proof, we transform any progressing trace of the original proof into a progressing trace after transformation.
    We do so by defining a progressing trace for each infinite branch after each transformation step given a progressing trace of the proof before the transformation step.

    Consider some infinite branch $\branch$ with progressing trace $(\z_i)_{i < \omega}$ of $\pi$ where $\Seq_l$ precedes a lowest $m\subV$ application.
    We define a progressing trace $(\z'_i)_{i < \omega}$ for the branch after a single permutation step occurring at $\Seq_l$ as follows.
    
    Case 1: If the label $\z_l$ of the trace at $\Seq_l$ disappears due to substitution before changing directly to another label which is not substituted, we have $\z_l = \y_k$ for some $k \leq n$ and $\z_{l+1} \neq \y_j$ for all $j \leq n$.
    Then, let $\z_l' := \z_{l+1} =: \z_{l+1}'$, thus the new trace changes from the label $\z_{l-1} = \x_k$ to $\z_{l+1}$ directly which is possible as $\z_{l+1}$ is not being substituted, thus also satisfying condition (1) from Definition~\ref{def:trace} for the instance of $m\subV$.
    
    Case 2: If the label $\z_l$ of the trace at $\Seq_l$ disappears due to substitution before changing to another label which is substituted, we have $\z_l = \y_k$ and $\z_{l+1} = \y_{k'}$ for some $k,k' \leq n$.
    Then, let $\z_l' := \x_{k'}$ and $\z_{l+1}' := \z_{l+1} = \y_{k'}$, thus the new trace goes from the label $\z_{l-1} = \x_k$ to $\x_{k'}$ before becoming $\y_{k'}$ due to the substitution.

    Case 3: If the label $\z_l$ of the trace at $\Seq_l$ does not disappear due to substitution, we have $\z_l \neq \y_k$ for all $k \leq n$ as well as $\z_{l-1} = \z_l$. Then, let $\z_l' := \z_{l+1} =: \z_{l+1}'$.

    Further define $\z_i' := \z_i$ for all $i$ whenever $i<l$ or  $l+1 < i$.

    Note that at each step no progress point is moved upwards along the branch $\branch$, thus by corecursion any infinite branch we obtain after the full transformation must have a trace which is progressing. 
    Thus, we have obtained a proof of $\Seq$ in $\labCSX\setminus\{\subV\}$.
\end{proof}

\begin{lem} \label{lem:wk-sub-elim}
    The rule $\wk$ is eliminable in $\labCSX \setminus \{\subV \}$.
\end{lem}
\begin{proof}

    Let $\pi$ be a proof of a sequent $\Seq$ in $\labCSX\setminus \{\subV \}$. 
    We begin by making all rule instances of non-structural and non-axiomatic rules in $\pi$ maximal.
    Formally, we obtain a proof $\pi' := \alpha(\pi)$ by applying a corecursive transformation $\alpha(\cdot)$ to $\pi$, where $\alpha$ commutes with all structural, axiomatic and maximal rule instances.
    For non-maximal rule instances, $\alpha$ is defined as follows:
    
    $\toR$:
    \[\alpha\left( 
    \vlderivation{
    \vlin{\toR}{}{\R, \Gamma \Ra \Theta , \x : A \to B}{
    \vltrind{\pi}{\R, \Gamma, \x : A \Ra \Theta ,  \x :  B }{4mm}
    }
    } \right) \qquad = \qquad
    \vlderivation{
    \vlin{\toR}{}{\R, \Gamma \Ra \Theta , \x : A \to B }{
    \vlin{\wk}{}{\R, \Gamma, \x : A \Ra \Theta , \x :  B , \x : A \to B}{
    \vlhy{\alpha\left(\vls\vlderivation{\vltrind{\pi}{\R, \Gamma, \x : A \Ra  \Theta , \x :  B }{4mm}} \right)}
    }
    }
    }
    \]

    $\toL$, $\BoxR_i$, $\Box_i{\Ra}$, tr$_{R_iR_j}$, $\Mont$, $\Pione$, $\er$: Similar. 

    To preserve the validity of the proof $\pi'$, change all progressing traces as follows. Whenever the trace reaches an instance of weakening which was added due to $\alpha$, add an additional instance of the previous label at that point of the trace, and thus postpone the remainder of the trace.
    This leaves the progressing traces well-defined and progressing as labels that were initially in the trace occur at the same sequent as before. 
    Furthermore, every progress point is only delayed by finitely many steps.
    
    We define a stepwise corecursive proof transformation of $\pi'$ to obtain a proof of $\Seq$ without the rules $\wk$ and $\subV$ present.
    
    Consider some lowest instance of $\wk$ and permute that instance upwards, depending on the shape of the preceding rule:

    $\Id$:
    \[
    \vlderivation{
    \vlin{\wk}{}{\R, \R' , \Gamma, \Gamma' , \x : p \Ra  \Theta, \Theta', \x : p}{
    \vlin{\Id}{}{\R , \Gamma, \x : p \Ra \Theta, \x : p}{\vlhy{}}
    }
    } \quad \rightsquigarrow \quad 
    \vlinf{\Id}{}{\R, \R' , \Gamma, \Gamma' , \x : p \Ra  \Theta, \Theta', \x : p}{}
    \]

    $\botL$: Similar.

    $\toR$:
    \[
    \vlderivation{
    \vlin{\wk}{}{\R, \R', \Gamma , \Gamma' \Ra  \Theta , \Theta',\x : A \to B }{
    \vlin{\toR}{}{\R, \Gamma \Ra  \Theta , \x : A \to B  }{
    \vltrind{\pi'}{\R, \Gamma , \x : A \Ra \Theta, \x :  B }{4mm}
    }
    }
    }\quad \rightsquigarrow \quad 
    \vlderivation{
    \vlin{\toR}{}{\R, \R', \Gamma , \Gamma' \Ra  \Theta , \Theta', \x : A \to B }{
    \vlin{\wk}{}{\R, \R', \Gamma , \Gamma' , \x : A \Ra \Theta , \Theta', \x : B }{
    \vltrind{\pi'}{\R, \Gamma , \x : A \Ra\Theta , \x :  B }{4mm}
    }
    }
    }
    \]

    $\toL$, $\Box_i{\Ra}$, tr$_{R_iR_j}$, $\Mont$, $\Pione$, $\er$: Similar.

    $\BoxR_i$: Note that in this case the fresh label $\y$ might be in $\R' , \Gamma'$ or $\Theta'$, which is why we require a possible instance of substitution.
    \[
    \vlderivation{
    \vlin{\wk}{}{\R , \R' , \Gamma, \Gamma'  \Ra  \Theta, \Theta' , \x : \Box_i A}{
    \vlin{\BoxR_i}{}{\R , \Gamma \Ra  \Theta ,\x : \Box_i A }{
    \vltrind{\pi'}{\R , \x R_i \y , \Gamma \Ra  \Theta, \y : A }{4mm}
    }
    }
    } \quad \rightsquigarrow \quad 
    \vlderivation{
    \vlin{\BoxR_i}{}{\R , \R' , \Gamma, \Gamma'  \Ra  \Theta, \Theta' , \x : \Box_i A }{
    \vlin{\wk}{}{\R , \R' , \x R_i \z , \Gamma, \Gamma'  \Ra \Theta, \Theta' , \z : A }{
    \vlin{\subV}{}{\R , \x R_i \z , \Gamma \Ra  \Theta , \z : A }{
    \vltrind{\pi'}{\R , \x R_i \y , \Gamma \Ra \Theta ,  \y : A }{4mm}
    }
    }
    }
    } 
    \]
    By Lemma~\ref{lem:sub-elim} we can transform $\pi'$ into a proof without $\subV$:
    \[
    \vlderivation{
    \vlin{\BoxR_i}{}{\R , \R' , \Gamma, \Gamma'  \Ra \x : \Box_i A , \Theta, \Theta'}{
    \vlin{\wk}{}{\R , \R' , \x R_i \z , \Gamma, \Gamma'  \Ra \z : A , \Theta, \Theta'}{
    \vltrind{\Tilde\pi'}{\R , \x R_i \z , \Gamma \Ra \z : A , \Theta}{4mm}
    }
    }
    } 
    \]
    
    $\wk$:
    \[
    \vlderivation{
    \vlin{\wk}{}{\R , \R' , \R'', \Gamma, \Gamma' , \Gamma'' \Ra \Theta, \Theta', \Theta''}{
    \vlin{\wk}{}{\R , \R' ,  \Gamma, \Gamma'  \Ra \Theta, \Theta'}{
    \vltrind{\pi'}{\R ,  \Gamma \Ra \Theta}{4mm}
    }
    }
    } \rightsquigarrow \quad 
    \vlderivation{
    \vlin{\wdh}{}{\R , \R' , \R'', \Gamma, \Gamma' , \Gamma'' \Ra \Theta, \Theta', \Theta''}{
    \vlin{\wk}{}{\R , \R' , \R'', \Gamma, \Gamma' , \Gamma'' \Ra \Theta, \Theta', \Theta''}{
    \vltrind{\pi'}{\R ,  \Gamma \Ra \Theta}{4mm}
    }
    }
    } 
    \]

    As each transformation step is productive or terminating, we obtain a pre-proof without any instances of weakening.

    Each non-terminating permutation step, except for $\BoxR_i$, strictly increases the size of sequents.
    For $\BoxR_i$ we take the progressing trace given by Lemma~\ref{lem:sub-elim}.
    Thus any progressing trace from before the transformation is either also a progressing trace in the final proof or can be transformed into one according to Lemma~\ref{lem:sub-elim}. 
    Thus, we obtained a proof of $\Seq$ in $\labCSX\setminus\{\subV,\wk\}$.
\end{proof}

It is left to show that we can also eliminate instances of $\wdh$ in proofs. 
However, as the proof in Figure~\ref{fig:cyclicSequentProof} shows, this is not possible in general.
Thus, we need to restrict ourselves to sequents which are well behaved and show that they are sufficient for proving formulas, which we do in the following.

\begin{figure}
    \centering
    \[
    \vlderivation{
    \vlin{\wdh}{\bullet}{\x R \x \Ra \x : \bot}{
    \vlin{\wdh}{\bullet}{\x R \x \Ra \x : \bot}{\vlhy{}}
    }
    }
    \]
    \caption{Proof of a cyclic sequent in $\labCS$, lines denoted by $\bullet$ are identified forming a cycle in the proof}
    \label{fig:cyclicSequentProof}
\end{figure}

Let $\R^\ast$ and $\R^+$ be respectively the reflexive transitive and transitive closure of the set of all pairs of labels $(\x,\y)$ such that $\x R_i \y \in \R$ for some $i \in\{0,1\}$.
We call a sequent $\R, \Gamma \Ra \Theta$ \emph{tree-like} if the set of relational atoms $\R$ forms a rooted tree, i.e.~there is a label $\x_0$ such that for all $\y \in \Lab(\R, \Gamma \Ra \Theta)$ we have $(\x_0 ,\y) \in \R^\ast$ (rooted), for all $\x ,\y \in \Lab(\R, \Gamma \Ra \Theta)$ we have if $(\x, \y) \in \R^+$ then $(\y, \x) \notin \R^+$ (anti-symmetric and irreflexive), and for every $\x \in \Lab(\R, \Gamma \Ra \Theta)$ with $\x_0 \neq \x$ there is exactly one $\y \in \Lab(\R, \Gamma \Ra \Theta)$ such that $(\y, \x) \in \R^+$ and whenever there is a $\z \neq \y$ with $(\z, \x) \in \R^+$, then also $(\z, \y) \in \R^+$ (every label, except the root has exactly one direct ancestor)
Note that the relation sets of tree-like sequents are moreover acyclic, which will be important for controlling the progress condition.
We call a sequent $\R ,\Gamma \Ra \Theta$ \emph{connected} if for any labels $\x , \y \in \Lab(\R , \Gamma \Ra \Theta)$ we have $(\x,\y) \in \R^\ast$.\footnote{It would be sufficient for our arguments to consider sequents which only have finitely many disconnected parts in $\R$.
However, as such sequents, similar to non-tree-like sequents, never occur in proofs, we work with the simpler property of connectedness.}

\begin{lem} \label{lem:tree-like}
    Let $\pi$ be a (pre-)proof in $\labCSX \setminus\{\subV,\wk\}$ of a connected tree-like sequent where every rule application is maximal, then every sequent occurring in $\pi$ is tree-like and connected. 
\end{lem}
\begin{proof}
    We proceed by induction over the structure of $\pi$: assuming that the conclusion of a maximal rule instance has a tree-like relation set $\R$ and is connected, we show that all premisses of that rule instance must be tree-like too.
    It suffices to consider the rules which actually change $\R$. 
    
    For tree-likeness, maximal instances of extension rules ($\Mont$ and $\Pione$) and transitivity rules ($\transRR, \transSS , \transRS, \transSR$), only add relational atoms for two labels which were already connected via an existing path.
    Further, maximal rule instances of $\Box_i{\Ra}$ and $\er$ leave the set of relational atoms unchanged.
    Thus, these rules leave $\R^\ast$ and $\R^+$ unchanged.

    For connectedness, as maximal rule instances do not remove formulas or relational atoms from sequents all labels in the conclusion must remain connected in the premiss. 
    Thus, it suffices to consider rules which introduce fresh labels.

    Consider a maximal instance of $\BoxR_i$ 
    \[
    \vlinf{\BoxR_i}{}{\R ,\Gamma \Ra \Theta , \x : \Box_i A}{\R , \x R_i \y , \Gamma \Ra \Theta , \y : A}
    \]
    which introduces a fresh label $\y$ with a relational atom $\x R_i \y$, where $\R ,\Gamma \Ra \Theta , \x : \Box_i A$ is tree-like and connected.
    Define $\Tilde{\R}:= \R, \x R_i \y$.
    As we must have $(\x_0,\x) \in \R^\ast $ where $\x_0$ is the root of $\R$, we must also have $(\x_0, \y) \in  \Tilde{\R}^\ast$.
    Further, as $\y$ was fresh we cannot have $(\y, \z) \in \Tilde{\R}^+$ for some $\z$.
    Finally, consider some $\z \neq \x$ with $(\z, \y) \in \Tilde{\R}^+$, then we must have that $(\z, \x) \in  \Tilde{\R}^+$ as the only relational atom in $\Tilde{\R}$ containing $\y$ is $\x R_i \y$.
    Thus, $\R , \x R_i \y , \Gamma \Ra \Theta , \y : A$ is tree-like.
    Consider some label $\z$ from $\R , \x R_i \y , \Gamma \Ra \Theta , \y : A$.
    If $\z =\y$ we also have $(\z, \y) \in {\Tilde\R}^\ast$.
    Otherwise, $\z$ also occurs in the conclusion of $\BoxR_i$ and thus we have $(\x, \z) \in \R^\ast$ by the conclusion being connected.
    This gives us $(\z , \y) \in {\Tilde\R}^\ast$, showing that $\R , \x R_i \y , \Gamma \Ra \Theta , \y : A$ is also connected since $\z$ was arbitrary.
\end{proof}

From now on, we will generally assume that all sequents we work with are tree-like and connected.
This is possible as proofs of formulas (i.e.\ sequents of the form $\Ra \x : A$) have a tree-like sequent at their root, and thus, as long as the proof does not use $\wk$ or $\subV$, it only contains tree-like sequents, by the previous lemma.
The following example shows that there are undesirable proofs of sequents which are not tree-like.

\begin{exa}\label{ex:cyclic-sequent}
    Consider the proof in Figure~\ref{fig:cyclicSequentProof} of the sequent $\x R \x \Ra \x : \bot$, whose relational set is cyclic and thus also not tree-like. 
    The proof is valid due to $\x^\omega$ being a progressing trace for the only branch of this proof, for which condition 2(a) from Definition~\ref{def:trace} holds for every sequent in the branch.
\end{exa}

\begin{lem} \label{lem:sub-trace} 
    Given a proof in $\labCSX\setminus \{\subV, \wk\}$ (containing only tree-like sequents) with an infinite progressing branch $\branch$.
    Then, there are infinitely many occurrences of ${\Ra}\Box_i$ in $\branch$.
\end{lem}
\begin{proof}
    Suppose there are only finitely many instances of $\BoxR_i$ in $\branch$.
    As we do not have substitution present, there are only a finite number of labels that can occur in $\branch$.
    However, as we must have a trace $\trace$ which is progressing there must be infinitely many instances of $\x_n R_i \x_{n+1}$ in $\branch$, which is only possible with finite labels if there is eventually a sequent in the branch whose relational set contains a cycle.
    This contradicts Lemma~\ref{lem:tree-like}, showing that we must have infinitely many instances of $\BoxR_i$ in $\branch$.
\end{proof}

\begin{lem} \label{lem:structural-elim}
    All structural rules $\subV, \wk $ and $ \wdh$ are eliminable in $\labCSX$.
\end{lem}
\begin{proof}
    Let $\pi$ be a proof of $\Seq$ in $\labCSX$.
    Apply Lemmas~\ref{lem:sub-elim} and~\ref{lem:wk-sub-elim} to obtain a proof $\Tilde{\pi}$ of $\Seq$ in $\labCSX \setminus \{\subV, \wk\}$ with all rule instances being maximal.
    We remove all instances of $\wdh$ by immediately applying the preceding rule to the upper sequent of every instance of $\wdh$.
    Formally, we define a corecursive transformation $\beta(\cdot)$ on $\tilde\pi$ which permutes with any rule which is not $\wdh$, and
    \[\beta\left( 
    \vlderivation{
    \vlin{\wdh}{}{\R, \Gamma \Ra  \Theta}{
    \vltrind{\pi'}{\R, \Gamma \Ra \Theta}{4mm}
    }
    } \right) \qquad = \qquad
    \beta \left( \vlderivation{
   \vltrind{\pi'}{\R, \Gamma\Ra \Theta}{4mm}} \right)
    \]

    This transformation is productive on proofs in $\labCSX \setminus \{\subV, \wk\}$ as every infinite and progressing branch must contain infinitely many rule instances which are not $\wdh$ by Lemma~\ref{lem:sub-trace}.
    
    The (progressing) traces will stay valid as all sequents along a branch are strictly increasing due to the fact that we do not have weakening present and all rule instances are maximal.
\end{proof}

It follows from this lemma that the system $\labCSXstr :=\labCSX \setminus \{\subV, \wk ,\wdh\}$ is equivalent to $\labCSX$ over tree-like sequents.

\subsection{Regular Proofs}

We consider regular proofs, which are finitely representable non-wellfounded proofs.

\begin{defi}
    We call a proof \emph{regular} if it has only finitely many distinct sub-proofs. 
\end{defi}

Equivalently, we can say that a proof is regular if it can be represented as a finite tree with possible back-links, which is why they are also referred to as cyclic proofs. 
If we are given a cyclic proof $\pi$ of a sequent $\Seq$, we also write $\vdashcy \Seq$ instead of $\vdash \Seq$.

We can regularise proofs of $\labCSX$ by finding repeats in a progressing branch. 
For a (possibly infinite) branch $\branch$ in a proof without substitution we say that a label $\x$ is \emph{introduced by the formula $\Box_i A$} if there is an instance of $\BoxR_i$ in that branch with $\x : A$ as auxiliary formula.

\begin{lem} \label{lem:repeats}
    Given a proof $\pi$ in $\labCSXstr$ (containing only tree-like sequents) with an infinite branch $\branch$ and a progressing trace $\trace$, there are at least two distinct labels $\x_n$ and $\x_m$ in $\trace$ which are introduced by the same formula $\Box_i A$.
\end{lem}
\begin{proof}
    By Lemma~\ref{lem:sub-trace}, there must be infinitely many instances of $\BoxR$ and $\triagR$ in $\branch$.
    Note that the root sequent $\Seq_0$ contains finitely many formulas and thus, by analyticity of the rules, only finitely many distinct formulas (up to renaming of labels) can occur in $(\Seq_i)_{i < \omega}$. 
    Therefore, there must also be a finite number of distinct modal formulas $\Box_i A$ which can be principal for $\BoxR$ and $\triagR$. 
    As there are infinitely many applications of these rules, there must be (by pigeonhole principle) some labels $\x_n$ and $\x_m$ which are freshly introduced by rule applications of the same principal formula $\Box_i A$ (again, up to renaming of labels). 
    Thus these labels are introduced by a rule application of the same principal formula. 
\end{proof}

\begin{prop} \label{prop:NWFisCyc}
    Let $\Seq$ be a tree-like sequent. $\labCSX \vdash \Seq$ iff $\labCSX \vdashcy \Seq$.
\end{prop}
\begin{proof}
    $(\Leftarrow)$ 
    Follows immediately, as any regular proof is already a non-wellfounded proof.

    $(\Ra)$
    Let there be a non-wellfounded proof $\pi$ of $\Seq$ in $\labCSX$.
    By Lemma~\ref{lem:structural-elim}, obtain a proof $\pi'$ of $\Seq$ in $\labCSXstr$, i.e.\ without occurrences of structural rules, in which all rule instances are maximal.
    We transform $\pi'$ into a regular proof by first extracting a finite derivation $\tree$ whose open leaves are occurrences of sequents which also occur somewhere below in the same branch. 
    Thus, secondly, we can unravel all open branches simultaneously by corecursively defining their continuations via the derivation of the previous instance of the same sequent in the same branch.

    By Lemma~\ref{lem:repeats} we know that for every branch $\branch$ in $\pi'$ there are two distinct labels $\x_n$ and $\x_m$ which are introduced by the same formula in $\branch$.
    Let us call the sequent $\widetilde\Seq$ in which both labels occur for the first time the \emph{fold point} of that branch.
    As for every branch in $\pi'$ there must be a fold point after a finite number of rule applications, we can prune $\pi'$ into a finite derivation $\tree$ where every open branch ends in a fold point (by weak K\H{o}nig's Lemma).

    Consider an arbitrary open leaf (i.e.\ a fold point) $\widetilde\Seq$ in $\tree$ where $\x$ and $\x'$ are labels which have been introduced by the same formula $\Box_i A$.
    We thus may also write $\widetilde{\Seq} := \R , \y R_i \x , \Gamma \Ra \Theta , \x : A$.
    As $\x'$ also gets introduced in the branch $\branch$ we find a sequent $\Seq' = \R' , \z R_i \x' , \Gamma' \Ra \Theta' , \x' : A$ below $\widetilde{\Seq}$ where $\x'$ has been first introduced.
    By the absence of structural rules and maximal rule instances, we also have $\Seq' \subseteq \widetilde{\Seq}$.
    Moreover, we have $\R' \subseteq \R$, $\Gamma' \subseteq \Gamma$ and $\Theta' \subseteq \Theta$.
    Further, as $\x$ and $\x'$ are distinct labels in a trace, we also have $(\x' , \x) \in \R^+$ (recall that $\R^+$ is the transitive closure of the relation underlying the set of relational atoms $\R$).

    We construct a derivation of the fold point $\widetilde{\Seq}$ from $\Seq'$ as follows:
    \[
    \vlderivation{\vliq{trans}{}{\widetilde{\Seq}}{\vlin{\wk}{}{{\widetilde{\Seq}}^{tr}}{\vlin{\subV}{}{{\Seq'}[\x / \x']}{\vlhy{{\Seq'}}}}}}
    \]
    Here, $\subV$ substitutes the label $\x$ for $\x'$ and $trans$ is an application of all possible transitivity rules (possibly none) closing ${\Seq'}$ under the rules $\transRR$, $\transRS$, $\transSR$ and $\transSS$.
    It is left to show that the derivation is well-defined by showing that ${\Seq'}[\x / \x'] \subseteq {\widetilde{\Seq}}^{tr}$, making the step of weakening in the middle possible. 

    As $\x'$ was freshly introduced in ${\Seq'}$ we have ${\Seq'}[\x / \x'] = \R' , \z R_i \x , \Gamma' \Ra \Theta' , \x : A$. 
    Thus, due to $\Gamma' \subseteq \Gamma$ and $\Theta' \subseteq \Theta$, all labelled formulas in ${\Seq'}[\x / \x']$ must also be in ${\widetilde{\Seq}}^{tr}$. 
    Similarly, we can show that any relational atom in ${\Seq'}[\x / \x']$ is also in ${\widetilde{\Seq}}^{tr}$: 
    By $\R' \subseteq \R$ it suffices to show that $\z R_i \x$ occurs in $\widetilde{\Seq}^{tr}$.
    By $\z R_i \x'$  occurring in $\Seq'$ and $(\x' ,\x) \in \R^+$ there is $\y R_i \x$ in ${\widetilde{\Seq}}^{tr}$, due to closure under all transitivity rules. 
    This shows ${\Seq'}[\x / \x'] \subseteq {\widetilde{\Seq}}^{tr}$, making this derivation well-defined. 

    Applying such a derivation to all (finitely many) fold points in $\tree$ yields a derivation $\tree'$ whose open leaves are all repeats of sequents that occur below in the same branch.
    Thus, we can make this derivation into a cyclic pre-proof in the obvious way.
    Formally, we define this pre-proofs via a corecursive transformation $\tau(\tree')$ as follows over the leaves of $\tree'$.
    
    \[
    \vltreeder{\tree'}{{\Seq}}{\tau({\Seq}_1')}{\dots}{\tau({\Seq}_n')}
    \rightsquigarrow
    \toks0={0.7}
    \vlderivation{
    \vltrf{\tree'}{{\Seq}}{
        \vltr{\tree_1}{{\Seq}_1'}{\vlhy{\tau({\Seq}_1')}}{\vlhy{\dots}}{\vlhy{\tau({\Seq}_n')}}
    }{\vlhy{\dots}}{
        \vltr{\tree_n}{{\Seq}_n'}{\vlhy{\tau({\Seq}_1')}}{\vlhy{\dots}}{\vlhy{\tau({\Seq}_n')}}
    }{\the\toks0}
    }
    \]
    The sequents ${\Seq}_k'$ are the open leaves of $\tree'$ with the derivations $\tree_k$ being the subtrees of $\tree$ with the lower instances of $\Seq_k'$ as their root, possibly themselves containing fold points.

    It is left to show that $\tau(\tree')$ is a proof by showing that all its infinite branches are progressing.
    Consider any infinite branch $\branch$ in $\tau(\tree')$.
    By the infinite pigeonhole principle, this branch must pass through some $\widetilde\Seq_k$ and $\Seq_k'$ infinitely often.
    Let $\x_k$ and $\x_k'$ be the two labels introduced via the same formula by virtue of Lemma~\ref{lem:repeats}; then, there is a progressing trace $\trace \in \{\x_k , \x_k'\}^\omega$, which can be constructed in the obvious way by progressing via $\x_k' R_i \x_k \in \Seq_k^{tr}$ and changing from $\x_k'$ to $\x_k$ via the label substitution.
\end{proof}

Note that our transformation of non-wellfounded proofs into cyclic ones does not depend on our extension rules. 
We therefore also expect this connection between non-wellfounded and regular proofs to hold for a wider class of extensions, which are not covered here.
While we have eliminability of structural rules for non-wellfounded proofs, we do not have this in the cyclic setting, especially as our construction of a cyclic proof from a non-wellfounded one introduced possible instances of structural rules.

\begin{figure}[t]
    \begin{adjustbox}{max width=\textwidth}
    \centering
    $
    \vlderivation{
    \vlin{\toR}{}{\Ra \x : \triangle (\triangle p \to p) \to \Box ( \neg \triangle \bot \to p)}{
    \vlin{\BoxR}{}{\x : \triangle (\triangle p \to p) \Ra  \x : \Box ( \neg \triangle \bot \to p)}{
    \vlin{\toR}{}{\x R \y , \x : \triangle (\triangle p \to p) \Ra  \y : \neg \triangle \bot \to p}{
    \vlin{\neg{\Ra}}{}{\x R \y , \x : \triangle (\triangle p \to p), \y : \neg \triangle \bot \Ra  \y : p}{
    \vlin{\triagR}{}{\x R \y ,  \x : \triangle (\triangle p \to p) \Ra  \y : p, \y : \triangle \bot}{
    \vlin{\Pione}{}{\x R \y , \y S \z , \x : \triangle (\triangle p \to p) \Ra  \y : p, \z : \bot}{
     \vlin{\wk}{}{\x R \y , \x S \y , \x : \triangle (\triangle p \to p) \Ra  \y : p, \z : \bot}{
    \vlin{\triagL}{\bullet}{ \x S \y , \x : \triangle (\triangle p \to p) \Ra  \y : p}{
    \vliin{\toL}{}{ \x S \y , \x : \triangle (\triangle p \to p), \y : \triangle p \to p \Ra  \y : p}{
        \vlin{\triagR}{}{ \x S \y, \x : \triangle (\triangle p \to p) \Ra  \y : p, \y : \triangle p}{
        \vlin{\transSS}{}{\x S \y , \y S \u , \x : \triangle (\triangle p \to p) \Ra  \y : p, \u : p}{
        \vlin{\wk}{}{\x S \u, \x : \triangle (\triangle p \to p) \Ra  \y : p, \u : p}{
        \vlin{\subV}{}{ \x S \u, \x : \triangle (\triangle p \to p) \Ra   \u : p}{
        \vlin{\triagL}{\bullet}{\x S \y, \x : \triangle (\triangle p \to p) \Ra   \y : p}{\vlhy{}}
        }
        }
        }
        }
    }{
        \vlin{\Id}{}{\x S \y , \x : \triangle (\triangle p \to p), \y : p \Ra  \y : p}{\vlhy{}}
    }
    }
    }
    }
    }
    }
    }
    }
    }
    }
    $
    \end{adjustbox}
    \caption{A cyclic proof in $\ell^\infty \CSMP$.}
    \label{fig:CSMPproofcyc}
\end{figure}

\begin{exa}
    Consider the cyclic proof in Figure~\ref{fig:CSMPproofcyc}. 
    Observe that a progressing trace can be formed from $\y$ and $\u$, as between the cyclically connected nodes, which are denoted by $\bullet$, we can progress from $\y$ to $\u$ via $S$ after which we go back from $\u$ to $\y$ via the substitution.
    Note that, while this cyclic proof proves the same formula as the proof in Figure~\ref{fig:CSMPproof}, it is not obtained via the transformation we have given in Proposition~\ref{prop:NWFisCyc}.
    We can see this as there is actually no pair of labels which are introduced by the same principal formula.
    For this to be the case, we would have to unravel the cyclic branch at least once to obtain a successor label of $\u$ which also gets introduced by the formula $\y : \triangle p$.
    Here, we instead added an additional weakening step above the application of $\Pione$ which allows us to cycle back earlier.
\end{exa}

\section{Completeness via Proof Search} \label{sect:completeness}

In this section, we provide a semantic completeness argument for the systems $\lCS{\A^r}$ (with $\A \in \{ \emptyset , \{M\}, \{M , P\}, \{ER\}$) a set of axioms and $\A^r$ the respective set of rules) via an infinite proof search method.
Formally speaking, we read countermodels directly from labelled sequents to which no effective rules can be applied.
While this is rather straightforward and standard for the frame-complete logics, we will need an additional construction to gain models for the logic $\ER$.
Together with the soundness of our system which is factored through the finitary system (see Theorem~\ref{thm:indSemanticSoundness} and Proposition~\ref{prop:NWFtoInd}), which we introduce in the following section, we obtain the following soundness and completeness result.

\begin{thm}
    $\lCS{\A^r} \vdash A$ iff $\CS\A \vdash A$. \qed
\end{thm}

For this, we interpret sequents on models as follows.

\begin{defi}[Sequent interpretation] \label{def:SequInterpret}
    Let $\Seq = \R , \Gamma \Ra \Theta$ be a sequent and $\Mw = \la W , \prec_0 , \prec_1 , V \ra$ a \Smo model. 
    Let $I : \Lab(\Seq) \to W$ be a function from labels of the sequent to points in the model.
    We call $I$ a \emph{sequent interpretation} of $\Seq$ on $\Mw$ if the following condition holds for all $\x,\y \in \Lab(\Seq)$: If $\x R_i \y \in \R$ then $I(\x) \prec_i I(\y)$.
\end{defi}

Similar to usual sequents $\Gamma \Ra \Theta$ which can be interpreted as formulas ${\bigwedge \Gamma \to \bigvee \Theta}$, we can interpret labelled sequents directly on models yielding a semantic notion of validity. 

\begin{defi}[Semantic Truth for Sequents] \label{def:SequTruth}
    A sequent $\Seq = \R , \Gamma \Ra \Theta$ \emph{holds} in a model $\Mw = \la W , \prec_0 , \prec_1 , V \ra$ under some interpretation $I$, written as $\Mw \Vdash_I \Seq$, if either $\Mw , I(\x) \nvDash A$ for some $\x : A \in \Gamma$ or $\Mw, I(\y) \vDash B$ for some $\y : B \in \Theta$. 
    We say that $\Seq$ is \emph{valid} on $\Mw$, written as $\Mw \Vdash \Seq$, if $\Mw \Vdash_I \Seq$ for all interpretations $I$ of $\Seq$ on $\Mw$.
\end{defi}

Note that one can also read the condition of semantic truth as ``whenever all formulas in the antecedent hold, there must be some formula holding in the succedent''.
For clarity, we distinguish validity of formulas $\vDash$ from validity of sequents $\Vdash$.
We also generalise the notion of model-theoretic validity of a logic $\Lo$ towards sequents, written $\Lo \Vdash \Seq$, which holds if $\Mw \Vdash \Seq$ for all $\Lo$-models $\Mw$.
Note that by semantic completeness of $\Lo$ this definition gives us that $\Lo \Vdash \Ra \x : A$ is equivalent to $\Lo \vdash A$.

We begin the completeness proof by defining a primitive proof search procedure via saturation. 
This procedure is comparable with canonical model constructions for modal logics. 
We call the procedure primitive as it only applies rules blindly as long as they are not redundant.

\begin{defi}[Saturation] \label{def:saturation}
    Let $\R , \Gamma \Ra \Theta$ be a sequent. We call a labelled formula $\x :A$ \emph{saturated} in $\R , \Gamma \Ra \Theta$ if one of the following conditions hold:
    \begin{itemize}
        \item[-] $A$ is a propositional atom $p \in \Prop$ or $A= \bot$; 
        \item[-] if $A= B \to C \in \Gamma$ then either $\x :B \in \Theta$ or $\x : C \in \Gamma$;
        \item[-] if $A= B \to C \in \Theta$ then $\x :B \in \Gamma$ and $\x : C \in \Theta$;
        \item[-] if $A = \Box_i B \in \Gamma$ then for all $\y$ with $\x R_i\y \in \R$ we have $\y :B \in \Gamma$ (for all $i \in \{0,1\}$);
        \item[-] if $A = \Box_i B \in \Theta$ then there is some $\y$ with $\x R_i\y \in \R$ and $\y :B \in \Theta$ (for all $i \in \{0,1\}$).
    \end{itemize}
Given a set of extension rules $\X \subseteq \{\Mont, \Pione, \er\}$, we call a non-axiomatic sequent $\R , \Gamma \Ra \Theta$ $\X$-saturated if every formula in it is saturated and for any labels $\x , \y , \z \in \Var$ the following holds:
  \begin{itemize}
        	\item[-] if $\x R_i \y , \y R_j \z \in \R$ then $\x R_j \z \in \R$ (for all $i , j\in \{0,1\}$);
\item[-] if $\Mont \in \X$ and $\x S \y \in \R$ then $\x R \y \in \R$;
\item[-] if $\Pione \in \X $ and $\x R \y, \y S \z \in \R$ then $\x S \y \in \R$;
\item[-] if $\er \in \X$ and $\x S \y \in \R$ then for every formula $A$,  if $\y : \Box A \in \Gamma$ then $\y : A \in \Gamma$.
    \end{itemize}

    Given a derivation $\tree$, we define its \emph{$\X$-saturation} by iterating the following step:
    Fix an arbitrary ordering of the rules in $\labCSXstr$\footnote{The ordering only ensures that every rule will eventually be applied if it can be applied.} and repeatedly cycle through this ordering, applying each rule maximally to every open leaf in $\tree$
    whenever possible and whenever no premiss of the rule instance is equal to the conclusion of that instance.
\end{defi}

Due to the maximality of rule instances, sequents are strictly increasing along branches. 
The reason for requiring that sequents strictly grow along branches under $\X$-saturation is that it will ensure that a countermodel for an upper sequent is a fortiori also a countermodel for the lower sequent.
Clearly, $\X$-saturation yields a derivation whose open leaves are all $\X$-saturated.
Crucially, we can read out countermodels from saturated sequents.

\begin{defi}[Induced model]
    Given an $\A^\r$-saturated sequent $\Seq=\R , \Gamma \Ra \Theta$, we define the \emph{model $\Mw_\Seq =\la W , \prec_0 , \prec_1 , V\ra$ induced by $\Seq$} as follows.
    \begin{itemize}
        \item $W := \Lab(\Seq)$
        \item ${\prec_i} := \{ (\x,\y) \mid \x R_i \y \in \R\}$ for each $i \in \{0,1\}$
        \item $V(p):= \{\x \mid \x: p \in \Gamma\}$
    \end{itemize}
\end{defi}

\begin{lem}[Truth Lemma] \label{lem:truth-lemma}
    Let $\Seq =\R , \Gamma \Ra \Theta$ be a saturated sequent and $\Mw_\Seq$ the induced model of $\Seq$. Then, for any labelled formula $\x :A$ the following hold:
    \begin{itemize}
        \item if $\x : A \in \Gamma$ then $\Mw_\Seq, \x \vDash A$, and
        \item if $\x : A \in \Theta$ then $\Mw_\Seq, \x \nvDash A$.
    \end{itemize}
\end{lem}
\begin{proof}
    We proceed by induction on $A$.

    Consider $A = \bot$, then we have $\x : \bot \notin \Gamma$ as $\Seq$ is saturated and thus not axiomatic.
    We further have $\Mw_\Seq , \x \nvDash \bot$ by definition.

    Consider $A = p$ with $\x : p \in \Gamma$, then by construction $\x \in V(p)$ and thus $\Mw_\Seq , \x \vDash p$.
    If $\x : p \in \Theta$ then $\x : p \notin \Gamma$ as $\Seq$ is non-axiomatic. Therefore, $\x \notin V(p)$ and thus $\Mw_\Seq , \x \nvDash p$.
    
    Consider $A = B \to C$ with $\x : B \to C \in \Gamma$, then by saturation either $\x : C \in \Gamma$ or $\x : B \in \Theta$, by inductive hypothesis then either $\Mw_\Seq ,\x \vDash C$ or $\Mw_\Seq,\x \nvDash B$ and thus $\Mw_\Seq , \x \vDash B \to C$.
    The case for $\x : B \to C \in \Theta$ is similar.
    
    Consider $A = \Box_i B$ with $\x : \Box_i B \in \Gamma$ and take $\y \in W$ such that $\x \prec_i \y$. 
    If no such $\y$ exists, we already have $\Mw_\Seq , \x \vDash \Box_i B$. 
    By construction of the model, we have $\x R_i \y \in \R$ and thus by saturation $\y : B \in \Gamma$.
    By inductive hypothesis, $\Mw_\Seq , \y \vDash B$, showing as $\y$ was arbitrary $\Mw_\Seq , \x \vDash \Box_i B$.
    If $\Box_i B \in \Theta$, we have by saturation that there are $\x R_i \y \in \R$ and $\y : B \in \Theta$.
    By inductive hypothesis, $\Mw_\Seq , \y \nvDash B$ and by construction of the model $\x \prec_i \y$, yielding $\Mw_\Seq , \x \nvDash \Box_i B$.
\end{proof}

As a corollary of the lemma, we immediately obtain that $\Mw_\Seq$ is a countermodel for $\Seq$.

\begin{cor} \label{cor:countermodel}
    For $\Seq$ a saturated sequent, we have $\Mw_\Seq \nVdash \Seq$.
\end{cor}
\begin{proof}
    Let $I_{id} : \Lab(\Seq) \to \Lab(\Seq)$ be the identity interpretation on the labels of $\Seq$, for which we then have $\Mw_\Seq \nVdash_{I_{id}} \Seq$.
\end{proof}

It is thus left to show that the induced model of an $\A^r$-saturated sequent is indeed also a model of the logic $\CSA$.
This does not directly apply to $\ER$ as the models are more complicated.
We thus consider the case of $\A = \{ER\}$ separately.

\begin{lem} \label{lem:model-of-CSMP}
    Let $\A \in \{\emptyset , \{M\} , \{M,P\}\}$ and let $\Seq$ be an $\A^\r$-saturated tree-like sequent, then the induced model $\Mw_\Seq$ is a model of the logic $\CS\A$.
\end{lem}
\begin{proof}
    Let $\Mw_\Seq$ be the induced model of $\Seq$. 
    If $\A = \emptyset$, due to saturation, $\Seq$ must be closed under the transitivity rules.
    Further, as the sequent is finite and tree-like, the model we obtained must be Noetherian and thus a model of $\CS$.
    For $\A=\{M\}$, the same argument holds while we additionally have $\Mont \in \A^\r$ and by $\A^\r$-saturation, $\R$ must be closed under $\Mont$, meaning that ${\prec_1} \subseteq {\prec_0}$, making $\Mw_\Seq$ a model of $\CSM$.  
    For $\A = \{M,P\} $, we have $\A^\r=\{\Mont , \Pione\}$, thus as before we have ${\prec_1} \subseteq {\prec_0}$, as well as $\prec_1$ being downwards closed, making $\Mw_\Seq$ a model of $\CSMP$. 
\end{proof}

For $\ER$, we first define a way to construct a proper model of $\ER$ from saturated sequents, which we call the \emph{extended model}.
For these models we also show that they are indeed models of $\ER$, as well as a separate truth lemma.

\begin{defi}[Extended model] \label{def:extended-ER-model}
    Let $\Seq$ be an $\{\Mont , \er\}$-saturated tree-like sequent and $\Mw_\Seq$ the induced model of $\Seq$. 
    We define an infinite sequence of models $(\Mw^k)_{k< \omega}$ with $\Mw^k =\la W^k , \prec_0^k , \prec_1^k, V^k\ra $ and $W^k$ being finite for all $k < \omega$ as follows.

    Define $\Mw_\Seq =: \Mw^0 =\la W^0 , \prec_0^0 , \prec_1^0, V^0\ra$. 
    Clearly, $W^0$ is finite by sequents being finite.
    
    Define $\Mw^1 :=\la W^1 , \prec_0^1 , \prec_1^1, V^1\ra$ as follows.
    Let $\{x_1 , ... , x_n\} \subseteq M_1^0$ be the points in $M_1^0$ (i.e.\ the set of $x$ s.t.\ $ y \prec_1^0 x$ for some $y \in W^0$) for which there is no $y \in W^0$ with $x \prec_0^0 y$.
    Let $Y = \{y_1 , ... , y_n\}$ be a set of fresh points, i.e.\ $\{y_1 , ... , y_n\} \cap W^0 = \emptyset$, then
    \begin{itemize}
        \item $W^1 := W^0 \cup Y$, 
        \item ${\prec_0^1} := {\prec_0^0} \cup \bigcup_{i \leq n} \{(z,y_i) \mid z \preceq_0^0 x_i \}$, 
        \item ${\prec_1^1} := {\prec_1^0}$, and 
        \item $V^1(p) = V^0(p) \cup \{y_i \mid x_i \in V^0(p)\}$ for all $p \in \Prop$. 
    \end{itemize}
    By adding only finitely many points, the set $W^1$ remains finite.

    Let $\Mw^k =\la W^k , \prec_0^k , \prec_1^k, V^k\ra$ be given for some $k > 0$, with $\{x_1 , ... , x_n\} = M_1^0 \subseteq W^k$ being all points in $M_1^0$.
    For any $x_i$, let $Y_i = \{y_{i1} , ... , y_{i m_i}\}$ be the set of the $m_i$ many direct successors of $x_i$ (i.e.\ $Y_i := \{y \mid x_i \prec_0^k y$ such that there is no $z$ with $ x_i \prec_0^k z \prec_0^k y\}$).
    Note that $Y_i$ is indeed finite as $Y_i \subseteq W^k$ and $W^k$ being finite by assumption.
    For every $i \leq n$ let $U_i = \{ u_{i11} , u_{i12}, ... , u_{i1m_i}, u_{i21} , ... , u_{i m_i m_i}\}$ be a set of $m_i^2$ many fresh points (i.e.\ $U_i \cap W^k = \emptyset$ and $U_i \cap U_j = \emptyset$ for $i \neq j$).
    Then, we define 
    \begin{itemize}
        \item $W^{k+1} := W^k \cup \bigcup_{i \leq n} U_i$, 
        \item ${\prec_0^{k+1}} := {\prec_0^{k}} \cup \{(z, u_{ijl}) \mid z \preceq_0^k x_{i}\} \cup \{(u_{ijl}, z) \mid y_{ij} \preceq_0^k z$ or $y_{il} \preceq_0^k z\}$, 
        \item ${\prec_1^{k+1}} := \{(v,w) \mid v\prec_0^{k+1}w $ and $w \in M_1^0\}$, and 
        \item $V^{k+1}(p) := V^k(p) \cup \{u_{ijl} \mid x_i \in V^0(p)\}$ for all $p \in \Prop$.
    \end{itemize}
    By adding only finitely many points, the set $W^{k+1}$ remains finite.

    We define the \emph{extended model} of $\Seq$ as $\Mw_\Seq^\ER := \la W^\omega , \prec_0^\omega , \prec_1^\omega, V^\omega\ra$.
    This limit is well-defined by the chain of models being strictly increasing.
\end{defi}

We say that the fresh points $U_i$ and $y_i \in Y$ are \emph{introduced by the point $x_i$}. 
Intuitively, the model $\Mw^1$ can be seen as the induced model closed under condition $(i)$ of $\ER$-frames as well as the transitivity conditions (see Definition~\ref{def:ER-semantics}), while $\Mw^{k+1}$ is obtained from $\Mw^k$ by closing under certain instances of condition $(ii)$ as well as the transitivity conditions on $\prec_0$ and $\prec_1$ (see also Figure~\ref{fig:LimitPoint}).
Note that in our construction, newly added points inherit the valuation from the points by which they are introduced.
Let us further observe that a point in $M_1^\omega := \{ x \in W^\omega \mid \exists y \in W^\omega. y \prec^\omega_1 x \}$ is per definition of $\prec_1^\omega$ also a point in $M_1^0$, i.e.\ $M_1^\omega = M_1^0$.

\begin{lem} \label{lem:ER-extended-model}
    For any $\{\Mont , \er\}$-saturated tree-like sequent $\Seq$ its extended model $\Mw_\Seq^\ER$ is a model of $\ER$.
\end{lem}
\begin{proof}
    We have due to the construction of $\prec_1^k$ that $M_1^0 = M_1^\omega$.
    We have by $\{\Mont\}$-saturation and construction that ${\prec_1^\omega} \subseteq {\prec_0^\omega}$, and both $\prec_0^\omega$ and $\prec_1^\omega$ are transitive.
    Thus, we also have ${\prec_1^\omega \circ \prec_0^\omega} \subseteq {\prec_0^\omega}$.

    We show that all points in $M_1^\omega$ are limit points (Definition~\ref{def:limit-point}). 
    Condition (i) follows directly from the construction of $\Mw^1$.
    For condition (ii), consider any $x \in M_1^\omega = M_1^0$ with any $y,z \in W^\omega$ such that $x \prec_0^\omega y$ and $x \prec_0^\omega z$. 
    Take some $k$ such that $x \prec_0^k y$ and $x \prec_0^k z$ and consider direct $\prec_0^k$-successors $y'$ and $z'$ of $x$ such that $y' \prec_0^k y$ and $z' \prec_0^k z$.
    This is possible, because the model $\Mw^k$ is finite.
    Then, by construction, there is some $u \in W^\omega$ such that $x \prec_0^{k+1} u$, $u \prec_0^{k+1} y'$ and $u \prec_0^{k+1} z'$.
    By transitivity of $\prec_0^\omega$, we also have $u \prec_0^\omega y$ and $u \prec_0^\omega z$, showing that condition (ii) holds.

It is left to show for the properties on the relations that ${\prec_0^\omega \circ \prec_1^\omega} \subseteq {\prec_1^\omega}$ and that $\prec_0^\omega$ is Noetherian.
For the former, assuming $x \prec_0^\omega y \prec_1^\omega z$ for any $x,y,z \in W^\omega$ we have already by $\Mont$ and transitivity $x \prec_0^\omega z$. 
As $z \in M_1^\omega$, we must also have $z \in M_1^0$ and thus by construction of $\prec_1^\omega$ we must eventually have $x \prec_1^{k} z$ for some $k < \omega$, and thus $x \prec_1^{\omega} z$.
To show Noetherianity assume for a contradiction that there is an infinite ascending chain in $\prec_0^\omega$. 
As the induced model was finite and acyclic, we must either have obtained a cycle due to our construction or the chain must consist of an infinite number of new points that we have added. 
The former, however, cannot happen as the construction (beyond $\Mw^1$) only adds relations with fresh points along already existing paths.
For the latter, note that no freshly added point can be a successor of a previously added point without both of them having a point from $M_1^\omega$ in between.
It thus follows that the chain must pass infinitely often through a point in $M_1^\omega$.
However, by $M_1^0 = M_1^\omega$ and $M_1^0$ being finite, we reach a contradiction.

This shows that the relation sets of $\Mw^\ER_\Seq$ fulfil the conditions of a (general) $\ER$-frame.
It is thus left to show that $V^\omega(p)$ is an admissible valuation for every $p \in \Prop$, in the sense of $V^\omega(p) \in \tau$ where $\la W^\omega , \prec_0^\omega , \prec_1^\omega , \tau \ra$ is the general frame obtained by the underlying frame of $\Mw^{\ER}_\Seq$ with $\tau$ the set of admissible valuations. 
We can construct a valuation $V(p)$ as follows from inclusive up- and strict downsets.
For points in the original induced model, we construct their singletons by $U_x := {x\!\uparrow_\preceq} \cap {x \! \downarrow_\prec} = \{x\}$.
To also capture the valuation of the newly introduced points (which inherit the valuation from the point by which they are introduced) we use the following combination which captures exactly these points for every point in $M_1^\omega$.
\[
V_x := x\!\uparrow_\preceq \setminus \left( \bigcup_{x \neq y \in W^0} y\!\uparrow_\preceq \right)
\]
Then, we can write the valuation as 
\[
V(p)= \bigcup_{ x \in V^0(p) \cap M_1^0} V_x \cup \bigcup_{ y \in V^0(p) \setminus M_1^0} U_y .
\]
To prove $V(p) = V^\omega(p)$, consider any point $x \in W^\omega$ for which we show that $x \in V(p) $ iff $x \in V^\omega(p)$.
It suffices to consider the case of $x \notin W^0$ as we have otherwise $x \in V^\omega(p)$ iff $x \in V^0(p)$, which immediately shows the claim.
For $x \in W^\omega \setminus W^0$ notice the following equivalence due to $U_y$ always being a singleton: $x \in V(p)$ iff there is some $y \in V^0(p) \cap M_1^0$ such that $x \in V_y$. 
$x \in V_y$ means that between $x$ and $y$ are no points from $W^0$, which is equivalent to $x$ being introduced via $y$ and by construction it follows that $x \in V^\omega(p)$.

As $W^0$ is finite (we started with a saturated sequent containing a finite number of labels) and all $U_x$ and $V_x$ are admissible sets, $V(p)$ must also be admissible, as it can be written as a combination of finite unions and complements. Thus, $\Mw^{\ER}_\Seq$ is an  $\ER$-model.
\end{proof}

To prove the truth lemma for extended models, we need to ensure that introduced points not only inherit the valuation of propositional atoms but of all formulas.
We obtain this from the following lemma.

\begin{lem} \label{lem:introduced-valuation}
    Let $\Mw_\Seq^\ER$ be an extended model and $x \in M_1^\omega$ a limit point, then $\Mw_\Seq^\ER , x \vDash A$ iff $\Mw_\Seq^\ER , y \vDash A$ holds for every $y$ introduced by $x$.
\end{lem}
\begin{proof}
   Proceed by induction on $A$. 
   $A = \bot$ follows immediately, whereas the case of $A=p$ follows from the construction as any fresh point added has the same valuation as the point it is introduced by.
   $A = B \to C$ follows directly via the inductive hypothesis.
   
   Assume $\Mw_\Seq^\ER , x \vDash \Box_i A$ and consider any $y$ introduced by $x$ and any $z$ such that $y \prec_i^\omega z$.
   By transitivity and $x \prec_1^\omega y$, we have $x \prec_i^\omega z$ and thus by assumption $\Mw_\Seq^\ER , z \vDash A$.
   This shows $\Mw_\Seq^\ER , y \vDash \Box_i A$.

   Assume $\Mw_\Seq^\ER , y \vDash \Box_i A$ holds for every $y$ introduced by $x$.
   Consider any $z$ such that $x \prec_i^\omega z$.
   By construction there is some $y$ introduced by $x$ such that $ y \prec_i^\omega z$.
   Thus, $\Mw_\Seq^\ER , z \vDash A$ which shows that $\Mw_\Seq^\ER , x \vDash \Box_i A$ as $z$ was chosen arbitrarily.
\end{proof}

\begin{cor} \label{cor:introduced-valuation}
    Let $\Mw_\Seq^\ER$ be an extended model with $x \in M_1^\omega$ a limit point such that $\Mw_\Seq^\ER , x \vDash A$ and $y$ a point which is introduced by $x$ then $\Mw_\Seq^\ER , y \vDash A$. \qed
\end{cor}

\begin{lem}[Truth Lemma for $\ER$] \label{lem:saturatedER}
    Let $\Seq = \R, \Gamma \Ra \Theta$ be an $\{\Mont , \er\}$-saturated sequent and $\Mw_\Seq^\ER$ the extended model of $\Seq$.
    Then we have for any labelled formula $\x :A$ the following
    \begin{itemize}
        \item if $\x : A \in \Gamma$ then $\Mw_\Seq^\ER, \x \vDash A$, and
        \item if $\x : A \in \Theta$ then $\Mw_\Seq^\ER, \x \nvDash A$.
    \end{itemize}
\end{lem}
\begin{proof}
We proceed by induction on $A$.
The base cases and implication are just like in the proof of Lemma~\ref{lem:truth-lemma}.

Let $A = \Box B$, and $\x : \Box B \in \Gamma$.
Consider any $y \in W^\omega$ such that $\x \prec_0^\omega y$.
If no such $y$ exists, we already have $\Mw_\Seq^\ER , \x \vDash \Box B$.
If $y \in W^0$, we immediately have $\x \prec_0^0 y$ and thus by construction of the model $\Mw^0$ we have $\x R y \in \R$.
By saturation we have $y : B$, from which we obtain by inductive hypothesis that $\Mw_\Seq^\ER , y \vDash B$.
If $y \notin W^0$, it must have been introduced at some stage $k$ of the model construction.
Thus, $y \in W^k$ with $\x \prec_0^k y$ as $\x$ was already in the induced model.
By construction, $y$ must have been introduced by some $\z \in M_1^{k-1} =M_1^0 \subseteq W^0$ with $\z \prec_0^k y$ and $\x \preceq_0^0 \z$.
If $\x = \z$, we have by construction $\u S \x \in \R$ for some $\u$, and by saturation over $\er$ we have $\x : B \in \Gamma$.
By inductive hypothesis then $\Mw_\Seq^\ER, \x \vDash B$ and by Corollary~\ref{cor:introduced-valuation} we have $\Mw_\Seq^\ER, y \vDash B$.
If $\x \prec_0^0 \z$, we have by construction $\x R\z \in \R$ and by saturation $\z : B \in \Gamma$.
By inductive hypothesis then $\Mw_\Seq^\ER, \z \vDash B$ and by Corollary~\ref{cor:introduced-valuation} $\Mw_\Seq^\ER, y \vDash B$.
As $y$ was chosen arbitrarily, we have shown that $\Mw_\Seq^\ER, \x \vDash \Box B$.

Let $A = \triangle B$, and $\x : \triangle B \in \Gamma$. 
Consider any $y$ such that $\x \prec_1^\omega y$. 
If no such $y$ exists, we already have $\Mw_\Seq^\ER , \x \vDash \triangle B$.
By $M_1^\omega = M_1^0$ we have $\x \prec_1^0 y$ and thus by construction of the model $\Mw^0$ we have $\x S y$.
By saturation we have $y : B$, from which we obtain by inductive hypothesis that $\Mw_\Seq^\ER , y \vDash B$.
As $y$ was arbitrary, we have $\Mw_\Seq^\ER , \x \vDash \triangle B$.

Let $A = \Box B$, and $\x : \Box B \in \Theta$.
By saturation, we have that there is some $\y$ with $\x R \y \in \R$ and $\y : B \in \Theta$.
By inductive hypothesis, $\Mw_\Seq^\ER , \y \nvDash B$, with $\x \prec_0^\omega \y$ due to construction and thus $\Mw_\Seq^\ER , \x \nvDash \Box B$.

$\x : \triangle B \in \Theta$ is similar.
\end{proof}

\begin{cor} \label{cor:ER-countermodel}
     For $\Seq$ a $\{\Mont, \er\}$-saturated sequent, we have $\Mw_\Seq^\ER \nVdash \Seq$. \qed
\end{cor}

Putting everything together, we can show that our saturation proof search will either allow us to obtain a countermodel or find a valid proof, thus providing semantic completeness.

\begin{thm} \label{thm:Completeness}
    Let $\Seq$ be a tree-like and connected sequent. If $\CS\A \Vdash \Seq$ then $\lCS{\A^r} \vdash \Seq$.
\end{thm}
\begin{proof}
    Assume $\CS\A \Vdash \Seq$, thus we have for every model $\Mw$ of $\CS\A$ and any interpretation $I$ that $\Mw \Vdash_I \Seq$ (by Theorem~\ref{thm:SemanticCompl}).
    Applying $\A^\r$-saturation to $\Seq$ gives us a possibly infinite derivation $\tree$. 
    
    For each of the branches of this tree, it must either be finite (1.) or infinite (2.). 
    We show that every finite branch must end in an axiomatic sequent and every infinite branch must be progressing, implying that $\tree$ is a proof. 

    (1.): Let $\Seq_n$ be the end sequent of a finite branch, which must be a saturated or axiomatic sequent. 
    Assume $\Seq_n$ is not an instance of $\Id$ or $\botL$. Then, by Lemma~\ref{lem:tree-like} the sequent $\Seq_n$ must be tree-like.
    By Corollary~\ref{cor:countermodel} and Lemmas~\ref{lem:model-of-CSMP} (or Corollary~\ref{cor:ER-countermodel} and Lemma~\ref{lem:ER-extended-model} if $\A = \{\er\}$) construct a countermodel $\Mw = \la W , \prec_0, \prec_1 , V\ra$ belonging to the respective logic with $\Mw \nVdash \Seq_n$. 
    Let $I : \Lab(\Seq_n) \to W$ be a sequent interpretation with $\Mw \nVdash_I \Seq_n$.
    As sequents in saturation are always expanding, we have $\Seq \subseteq \Seq_n$ and also $\Lab(\Seq) \subseteq \Lab(\Seq_n)$.
    Define $I' : \Lab(\Seq) \to W$ by $I'(\x) = I(\x)$ for all $\x \in \Lab(\Seq) \subseteq \Lab(\Seq_n)$ for which it is easy to check that $\Mw \nVdash_{I'} \Seq$.
    Thus $\Mw$ is also a countermodel for $\Seq$. 
    This contradicts our initial assumption that $\CS\A \Vdash \Seq$; therefore, $\Seq_n$ cannot be saturated and thus must be an axiomatic sequent.

    (2.): Let $(\Seq_i)_{i < \omega}$ be an infinite branch in $\tree$ with $\Seq_i =\R_i , \Gamma_i \Ra \Theta_i$ and let $\Seq_\omega = \R_\omega , \Gamma_\omega \Ra \Theta_\omega := \bigcup_{i < \omega} \R_i , \bigcup_{i < \omega} \Gamma_i \Ra \bigcup_{i < \omega} \Theta_i$ be the limit sequent of that branch.
    Suppose for a contradiction that $\Seq_\omega$ only contains finitely many labels.
    Then, there must be some $\Seq_i$ with $i < \omega$ such that $\Lab (\Seq_i) = \Lab(\Seq_\omega)$.
    As the number of labels cannot increase when saturating $\Seq_i$, it only requires a finite number of rule applications to saturate a formula in $\Seq_i$.
    As saturating all formulas in $\Seq_i$ only adds formulas of lower complexity, saturation must end in a finite number of steps.
    Further, $\X$-saturation is also obtainable in finitely many rule applications due to finitely many labels.
    This contradicts $\branch$ being an infinite branch of non-saturated sequents.
    Thus infinitely many labels occur in $\branch$ (by pigeonhole principle) and thus also in $\Seq_\omega$. 
    Note that every $\R_i$ is tree-like and connected via Lemma~\ref{lem:tree-like}; thus, as $\branch$ is a strictly increasing chain of sequents, $\Seq_\omega$ is tree-like and connected.
    Therefore, $\R_\omega$ is an infinite tree.
    Moreover, $\R_\omega$ is finitely branching, as a labelled modal formula $\x : \Box_j A$ in the succedent (of which there can only be finitely many) can be applied at most once.
    Thus, every label $\x$ in $\Seq_\omega$ has a finite number of direct successors $\y$, i.e.\ $\x R_j \y \in \R_\omega$ such that there is no $\z$ with $\x R_k \z , \z R_l \y \in \R_\omega$ for any $j,k,l \in \{0,1\}$.
    Thus, $\R_\omega$ is an infinite but finitely branching tree.
    Therefore, there is an infinite chain in $\R_\omega$ consisting of infinitely many distinct labels $(\x_k)_{k< \omega}$ s.t.\ $x_k R_j x_{k+1} \in \R_\omega$ by (weak) K\H{o}nig's Lemma. 
    This allows us to construct a progressing trace of $\branch$ in the obvious way, and it shows that $\tree$ is a proof in $\lCS{\A^r}$.
    Thus, $\lCS{\A^r}\vdash \Seq$. 
\end{proof}

\section{Inductive Proofs} \label{sec:inductive}

\begin{figure}
    \centering
        \[
        \vlinf{\BoxRind}{}{\R , \Gamma \Ra \Theta, \x : \Box A}{\R , \x R \y ,  \Gamma, \y : \Box A \Ra \Theta, \y :  A} \qquad
        \vlinf{\triagRind}{}{\R , \Gamma \Ra \Theta, \x : \triangle A}{\R , \x S \y ,  \Gamma, \y : \triangle A \Ra \Theta, \y :  A}
        \]
    \caption{Inductive rules for modalities in $\labCSXind$}
    \label{fig:IndRules}
\end{figure}

We define the system $\labCSXind$ (with $\X \subseteq \{\Mont, \Pione , \er\}$ a fixed set of extension rules) by the rules of $\labCSX$ where $\BoxR$ and $\triagR$ are replaced with $\BoxRind$ and $\triagRind$ respectively (see Figure~\ref{fig:IndRules}). 
For proofs in $\labCSXind$, we no longer allow for branches to be of infinite length.
Formally, we might say that proofs in $\labCSXind$ are those pre-proofs whose branches are all finite. 
Observe that this system subsumes Negri's labelled sequent calculus G3$\GL$ (Figure~\ref{fig:RulesG3GL}).

The main result of this section is the following soundness and completeness result of $\labCSXind$, where soundness is proved directly (Theorem~\ref{thm:indSemanticSoundness}), while completeness is factored through the non-wellfounded system (Theorem~\ref{thm:SemanticCompl} and Proposition~\ref{prop:NWFtoInd}).
\begin{thm}
    $\lCSind{\A^r} \vdash A$ iff $\CS\A \vdash A$. \qed
\end{thm}

We show directly via proof translations that the sequents provable in $\labCSX$ are also provable in $\labCSXind$.
This is done via a similar argument to recovering cyclic proofs from non-wellfounded ones (Proposition~\ref{prop:NWFisCyc}).

\begin{prop} \label{prop:NWFtoInd}
    Let $\Seq$ be a tree-like sequent. If $\labCSX \vdash\Seq$ then $\labCSXind \vdash \Seq$.
\end{prop}
\begin{proof}
    Let $\pi$ be a proof of $\Seq$ in $\labCSX$. 
    By Lemma~\ref{lem:structural-elim} we have a proof $\pi'$ of $\Seq$ in $\labCSXstr$ with maximal rule instances.
    As in the proof of Proposition~\ref{prop:NWFisCyc}, we have for every infinite branch $\branch$ in $\pi'$ by Lemma~\ref{lem:repeats} that there is a fold point $\Seq= \R ,\Gamma \Ra \Theta$ (i.e.\ there are two labels $\x_n$ and $\x_m$ which are introduced by the same principal formula $\Box_i A$ such that they first occur in $\Seq$).
    As these fold points occur after a finite number of rule applications, we can prune $\pi'$ into a finite derivation $\tree$ such that every open leaf is a fold point (by weak K\H{o}nig's Lemma).
    
    Replace in $\tree$ any occurrences of ${\Ra}\Box_i$ with ${\Ra}{\Box_i}^{ind}$, yielding a derivation $\tree'$ in $\labCSXind$.
    Formally, we define $\tree' =: \sigma_\emptyset(\tree)$ via a corecursive transformation $\sigma_{\Sigma}$. 
    $\Sigma$ is a set of labelled formulas by which we weaken the left context; this is necessary because of the additional diagonal formulas that are preserved upwards along branches.

    $\Id$:
    \[
    \sigma_\Sigma \left(
    \vlinf{\Id}{}{\R , \Gamma , \x : p \Ra \Theta , \x : p}{} 
    \right) \qquad = \qquad 
    \vlinf{\Id}{}{\R , \Sigma, \Gamma , \x : p \Ra \Theta , \x : p}{}
    \]

    $\botL$: Similar.

    $\toR$:
    \[
    \sigma_\Sigma \left(
    \vlderivation{
    \vlin{\toR}{}{\R , \Gamma \Ra \Theta, \x : A \to B}{
        \vltrind{\tree}{\R , \Gamma , \x : A \Ra \Theta,  \x : B}{4mm}
    }
    }
    \right) \qquad = \qquad
    \vlderivation{
    \vlin{\toR}{}{\R , \Gamma, \Sigma \Ra \Theta, \x : A \to B}{
    \vlhy{ \sigma_\Sigma\left( \vls\vlderivation{
        \vltrind{\tree}{\R , \Gamma , \x : A \Ra \Theta,  \x : B}{4mm}} \right)}
    }
    }
    \]

    $\toL, {\Box_i}{\Ra}, \trans_{R_i R_j}, \Mont , \Pione , \er$: Similar. 

    $\BoxR_i$:
    \[
    \sigma_\Sigma \left(
    \vlderivation{
    \vlin{\BoxR_i}{}{\R , \Gamma \Ra \Theta , \x : \Box_i A}{
    \vltrind{\tree}{\R , \x R_i \y , \Gamma  \Ra \Theta , \y : A}{4mm}
    }
    } 
    \right)
    \qquad = \qquad
    \vlderivation{
    \vlin{\BoxRind_i}{}{\R , \Sigma ,  \Gamma\Ra \Theta , \x : \Box_i A}{
    \vlhy{\sigma_{\Sigma \cup \{\y : \Box_i A \}}   \left(\vls \vlderivation{\vltrind{\tree}{\R , \x R_i \y , \Gamma  \Ra \Theta , \y : A}{4mm}} \right)}
    }
    }
    \]

    Open leaves:
    \[
    \sigma_\Sigma(\R, \Gamma \Ra \Theta ) \qquad = \qquad \R , \Sigma , \Gamma \Ra \Theta
    \]

    Note that $\tree'$ is a correct derivation in $\labCSXind$ as the contexts of premiss and conclusion are the same for all rule instances and we add the diagonal formula for instances of $\BoxRind_i$. 
    Further, all formulas are preserved in the sequents in which they occur. 
    
    Consider a branch $\finbranch$ in $\tree'$ with an open leaf $\Seq_n$. 
    As open leaves in $\tree$ are fold points and $\sigma_\Sigma$ preserves formula occurrences, we may write $\Seq_n = \R ,  \Gamma  \Ra \Theta , \x: A$, and we have that there is some label $\x'$ which gets introduced by $\Box_i A$ below $\Seq_n$ such that $(\x' , \x) \in \R^+$.
    Let $\Seq' = \R' , \Gamma'  \Ra \Theta' , \x' : A$ be the sequent in which $\x'$ occurs first in $\finbranch$.
    Due to the absence of structural rules and maximal rule instances, we have $\Seq' \subseteq \Seq_n$.
    We also have that $\x' : \Box_i A \in \Gamma' \subseteq \Gamma$ by $\Seq'$ being the premiss of a $\BoxRind_i$ instance with $\x' :A$ as auxiliary formula.

    Closing $\Seq_n$ under the transitivity rules (we denote this step by $trans$) yields a sequent whose relational atoms include $\x' R_i \x$ due to $(\x' , \x) \in \R^+$. 
    Thus, we can construct the following derivation:
    \[
    \vlderivation{
    \vliq{trans}{}{\R , \Gamma , \x' : \Box_i A \Ra  \Theta, \x : A}{
    \vlin{\Box_i{\Ra}}{}{\R^{tr} , \x' R_i \x, \Gamma , \x' : \Box_i A \Ra \Theta , \x : A }{
    \vlhy{\R^{tr} , \x' R_i \x, \Gamma , \x' : \Box_i A, \x : A \Ra  \Theta, \x : A}
    }
    }
    }
    \]
    The topmost sequent of this derivation is provable using a finite proof in $\labCS^{str}$ by Proposition~\ref{prop:GenID}.
    By transforming this proof via $\sigma_\emptyset$, we can prove the topmost sequent in $\labCSXind$.
    Doing this for all (finitely many) leaves of $\tree'$, we obtain a proof of $\Seq$ in $\labCSXind$.
\end{proof}

\begin{figure}[t]
\[
\vlderivation{
    \vlin{\triagRind}{}{\Ra \x : \triangle ( \Box \neg \Box \bot \to \bot)}{
    \vlin{\toL}{}{\x S \y , \y :  \triangle (\Box \neg \Box \bot \to \bot) \Ra \y :  \Box \neg \Box \bot \to \bot}{
    \vlin{\er}{}{\x S \y , \y :  \triangle (\Box \neg \Box \bot \to \bot) , \y :  \Box \neg \Box \bot \Ra  \y: \bot}{
    \vlin{\wk}{}{\x S \y , \y :  \triangle (\Box \neg \Box \bot \to \bot)  , \y :  \Box \neg \Box \bot, \y :  \neg \Box \bot \Ra  \y: \bot}{
    \vlin{\neg{\Ra}}{}{ \y :  \Box \neg \Box \bot, \y :  \neg \Box \bot \Ra }{
    \vlin{\BoxRind}{}{ \y :  \Box \neg \Box \bot  \Ra  \y : \Box \bot}{
    \vlin{\BoxL}{}{ \y R \z , \y :  \Box \neg \Box \bot , \z : \Box \bot \Ra   \z : \bot}{
    \vlin{\neg{\Ra}}{}{ \y R \z  , \z : \neg \Box\bot , \z : \Box \bot \Ra   \z : \bot}{
    \vlin{\BoxRind}{}{ \y R \z  , \z : \Box \bot \Ra   \z : \bot, \z : \Box\bot}{
    \vlin{\BoxL}{}{ \y R \z ,\z R \u   , \z : \Box \bot , \u :\Box \bot \Ra  \ \z : \bot, \u : \bot}{
    \vlin{\botL}{}{ \y R \z ,\z R \u  , \u : \bot , \u :\Box \bot \Ra \z : \bot, \u : \bot}{\vlhy{}}
    }
    }
    }
    }
    }
    }
    }
    }
    }
    }
    }
\]

    \caption{A proof in $\ell^{ind} \ER$.}
    \label{fig:ERproofind}
\end{figure}

\begin{exa}
    Consider the inductive proof in Figure~\ref{fig:ERproofind} of the bimodal formula expressing a strengthening of the second incompleteness theorem, namely that the stronger theory ``knows'' that it would be inconsistent if the weaker theory could prove its own consistency.
    Let us observe that this proof closes by means of both the diagonal formula $\z : \Box \bot$ as well as the extension rule of $\er$.
\end{exa} 

We show soundness of the system $\labCSXind$ via the respective \Smo semantics, thus also showing soundness of the other two previously introduced systems, and therefore completing our grand tour (Figure~\ref{fig:tour}).
For this, we interpret sequents over models just like we did in the completeness argument (Definition~\ref{def:SequInterpret}).

\begin{lem}[Validity of axiomatic sequents] \label{lem:ValAxSeq}
    If $\Seq$ is an axiomatic sequent, then $\CS\A \Vdash \Seq$ for any $\A \in \{\emptyset, \{M\},\{M,P\}, \{ER\}\}$.
\end{lem}
\begin{proof}
    We separately handle the cases on which kind of axiomatic sequent $\Seq$ is.

    $\Id$: Let $\Seq = \R , \Gamma , \x : p \Ra \Theta , \x :p $ be a sequent and $\Mw$ be any model $\Mw$ of $\CSA$. 
    It suffices to consider an interpretation $I$ of $\Seq$ which validates all antecedent formulas, i.e.\ $\Mw , I(\y) \vDash A$ for all $\y : A \in \Gamma \cup \{\x :p\}$.
    Then, we have $\Mw, I(x) \vDash p$ and thus there is a formula in the succedent of $\Seq$ which is true under this interpretation. This shows that $\Mw \Vdash \Seq$ for any model $\Mw$ of $\CSA$, and thus $\CSA \Vdash \Seq$.

    $\botL$: Let $\Seq = \R , \Gamma , \x : \bot \Ra\Theta$ be a sequent, $\Mw$ a model, and note that there is no interpretation $I$ validating $\Mw , I(\x) \vDash \bot$. Thus, $\Mw \Vdash \Seq$ holds vacuously, giving us $\CSA \Vdash \Seq$.
\end{proof}

\begin{lem}[Local Soundness]\label{lem:FalsityI}
    Let $\A \in \{\emptyset, \{M\},\{M,P\}, \{ER\}\}$, $\Mw$ be a model of $\CS\A$, and 
    \[
    \vlinf{\mathsf{r}}{(1 \leq n\leq 2)}{\Seq}{\Seq_1 \, ... \, \Seq_n}
    \]
    a rule instance from the rules of $\lCSind{\A^r}$ with at least one premiss. 
    If $\Mw \Vdash \Seq_i$ for all $i \leq n$ then $\Mw \Vdash \Seq$.
\end{lem}
\begin{proof}
    We show the statement contrapositively. 
    Thus, we assume that for a fixed model $\Mw = \la W , \prec_0 , \prec_1, V \ra$ we have an interpretation $I$ with $\Mw \nVdash_I \Seq$ and show that there is an interpretation $I_j$ with $\Mw \nVdash_{I_j} \Seq_j$ for some $j \leq n$.

    The proof is by inspection of the structure of the rule and the semantic meaning of the connectives as well as the conditions on \Smo models. 
    Consider the following key cases, as all others are straightforward or similar.

    ${\Ra}{\Box_i}^{ind}$: Let $\Seq = \R , \Gamma \Ra \Theta , \x : \Box_i A$ be a sequent and $I$ be a sequent interpretation with $\Mw \nVdash_I \R , \Gamma \Ra \Theta , \x : \Box_i A$. Therefore, for all  $\z :B \in \Gamma$ we have $\Mw, I (\z) \vDash B$, and for all $\z :B \in \Theta \cup \{ \x : \Box_i A\}$ we have $\Mw, I (\z) \nvDash B$. Furthermore, $\Mw, I(\x) \nvDash \Box_i A$. 
    Therefore, there is some $y_1 \in W$ with $I(\x) \prec_i y_1$ and $\Mw , y_1 \nvDash A$. 
    We claim that there is some $u \in W$ such that $I(\x) \prec_i u$, $\Mw, u \vDash \Box_i A$ and $\Mw, u \nvDash A$.
    Suppose for a contradiction that there is no such $u$.
    Then, we must have $\Mw, y_1 \nvDash \Box_i A$.
    We construct an infinite chain of points $y_1 \prec_i y_2 \prec_i \dots$ as follows: 
    given $y_n$ such that $\Mw, y_n \nvDash \Box_i A$ and $I(\x) \prec_i y_n$ let $y_{n+1}$ be some point such that $y_n \prec_i y_{n+1}$ and $\Mw, y_{n+1} \nvDash A$.
    By transitivity we have $I(\x) \prec_i y_{n+1}$, and thus by our assumption $\Mw, y_{n+1} \nvDash \Box_i A$, making the construction well-defined.
    This infinite chain contradicts the property of the relation $\prec_i$ being Noetherian, thus there must be some $u$ with  $I(\x) \prec_i u$, $\Mw, u \vDash \Box_i A$ and $\Mw, u \nvDash A$. 
    Define $I' : \Lab(\Seq)\cup\{\y\} \to W$ such that $I'(\v)= I(\v)$ for all $\v \in \Lab(\Seq)$ and $I' (\y)=u$ (note that $\y$ is fresh). Then, $\Mw , I'(\y) \nvDash A$, $\Mw , I'(\y) \vDash \Box_i A$ and therefore $\Mw \nVdash_{I'} \R , \x R \y , \Gamma , \y : \Box_i A \Ra  \Theta, \y : A$.
    
    $\Box_i{\Ra}$: Let $I$ be a sequent interpretation such that $\Mw \nVdash_I \R , \x R_i \y , \x : \Box_i A , \Gamma \Ra \Theta$. Then, for all $\z : B \in  \Gamma \cup \{\x : \Box_i A\}$ we have $\Mw , I(\z) \vDash B$. Furthermore, $\Mw , I(\x) \vDash \Box_i A$. 
    By $I$ being a sequent interpretation on $\Seq$ we have $I(\x) \prec_i I(\y)$, and therefore $\Mw , I(\y) \vDash A$. Clearly then also $\Mw \nVdash_I \R , \y : A , \Gamma \Ra \Theta$.

    $\transRS$: Let $I$ be a sequent interpretation such that $\Mw \nVdash_I \R , \x R \y , \y S \z , \Gamma \Ra \Theta$. By $\Mw$ being a \Smo model, as well as $I(\x) \prec_0 I(\y) $ and $ I(\y) \prec_1 I(\z)$, we have $I(\x) \prec_1 I(\z)$. Thus, we also have $\Mw \nVdash_I \R , \x S \z ,\Gamma \Ra \Theta$.

    $\subV$: Let $\Seq[\x /\y] = \R[\x /\y] , \Gamma[\x /\y] \Ra \Theta[\x /\y]$ be a sequent and $I$ be an interpretation such that $\Mw \nVdash_I \R[\x /\y] , \Gamma[\x /\y] \Ra \Theta[\x /\y]$.
    We have for arbitrary labels $\u$ and $\v$: if $\u R_i \v \in \R[\x /\y]$ then $I(\u) \prec_i I(\v)$; for all $\u : A \in  \Gamma[\x /\y]$ we have $\Mw , I(\u) \vDash A$; and for all $\u : A \in  \Theta[\x /\y]$ we have $\Mw , I(\u) \nvDash A$.
    Define $I' : \Lab(\Seq) \to W$ such that $I'(\z) = I(\z)$ for $\z \in \Lab(\Seq)\setminus\{\y\}$ and define $I'(\y) = I(\x)$.
    We then have for arbitrary labels $\u$ and $\v$: if $\u R_i \v \in \R$ then $I'(\u) \prec_i I'(\v)$; for all $\u : A \in  \Gamma$ we have $\Mw , I'(\u) \vDash A$; and for all $\u : A \in  \Theta$ we have $\Mw , I'(\u) \nvDash A$.
    Thus, $\Mw \nVdash_{I'} \R , \Gamma \Ra \Theta$.

    $\er$: Let $I$ be a sequent interpretation such that $\Mw \nVdash_I \R , \x S \y  , \Gamma , \y : \Box A \Ra \Theta$ (note that here $\Mw$ is an $\ER$-model), so we have $\Mw , I(\z) \vDash B$ for all $\z : B \in \Gamma \cup \{ \y: \Box A\}$, and therefore  $\Mw , I(\y) \vDash \Box A$. 
    By $I(\x) \prec_1 I(\y)$, $I(\y)$ must be a limit point, and so by Proposition~\ref{prop:limitPoint}: $\Mw , I(\y) \vDash \Box A \to A$ yields $\Mw , I(\y) \vDash A$ and therefore $\Mw \nVdash_I \R , \y : A ,\Gamma \Ra \Theta$.
\end{proof}

\begin{thm} \label{thm:indSemanticSoundness}
    If $\ell^{ind}\CS(\A^r) \vdash \Seq$ then $\Mw \Vdash \Seq$ for all models $\Mw$ of $\CS\A$.
\end{thm}
\begin{proof}
    Proceed by induction on the structure of the proof.
    For the base cases apply Lemma~\ref{lem:ValAxSeq}.
    In the inductive cases apply Lemma~\ref{lem:FalsityI}.
\end{proof}

\begin{cor}
    If $\ell^{ind}\CS(\A^r) \vdash A$ then $\CS\A \Vdash A$. \qed
\end{cor}

Note that unlike in the non-wellfounded and cyclic setting, we do not need to restrict soundness to tree-like sequents.
For example, the sequent $\x R\x \Ra \x :\bot$, for which we have a proof in $\labCS$ (Figure~\ref{fig:cyclicSequentProof}), is not provable in $\ell^{ind}\CS$.
As this is related to sequents which are not tree-like, it does not have any impact on provability of formulas in $\labCSXind$.

Let us finally point out two further consequences about our calculi which we obtain via our grand tour.
The first concerns the eliminability of structural rules in $\ell^{ind} CS(\A^r)$, the second a proof-theoretic description of the true provability logic $\GLS$.

\begin{cor} \label{cor:structural-eliminable}
    The structural rules $\wk$, $\subV$ and $\wdh$ are eliminable in $\lCSind{\A^r}$ (for proofs of tree-like and connected sequents).
\end{cor}
\begin{proof}
    Given a proof of a tree-like sequent $\Seq$ in $\lCSind{\A^r}$ with possible occurrences of $\wk$, $\subV$ and $\wdh$, we know that $\labCS(\A^r) \vdash \Seq$ (via Theorems~\ref{thm:indSemanticSoundness} and~\ref{thm:SemanticCompl}).
    Thus, by Lemma~\ref{lem:structural-elim} there is a proof of $\Seq$ in $\labCS^{str}(\A^r)$.
    Further, as the proof of Proposition~\ref{prop:NWFtoInd} does not rely on structural rules, there is a proof of $\Seq$ in $\lCSind{\A^r} \setminus\{\subV,\wk, \wdh\}$.
\end{proof}

Recall that we can describe the true provability logic $\mathsf{GLS}$ as all $\triangle$-theorems of $\ER$ (see Proposition~\ref{prop:type-ER-GLS}). 
Thus, we also have the following soundness and completeness result.

\begin{cor} \label{cor:ER-GLS}
    For all $A \in \LaBox$ the following are equivalent:
    \begin{enumerate}
        \item $\mathsf{GLS} \vdash A$
        \item $\ell^\infty\ER \vdash \triangle A$
        \item $\ell^\infty\ER \vdash \x S \y \Ra \y : A$ 
    \end{enumerate}
\end{cor}
\begin{proof}
    $(1) \Longrightarrow (2)$: Assume $\mathsf{GLS} \vdash A$. Then, by Proposition~\ref{prop:type-ER-GLS} we have $\ER \vdash \triangle A$, by Theorem~\ref{thm:SemanticCompl} we have $\Mw \vDash \triangle A$ for every $\ER$-model $\Mw$, and by Theorem~\ref{thm:Completeness} we have $\ell^\infty\ER\vdash \triangle A$.

    $(2) \Longrightarrow (3)$: Assume $\ell^\infty\ER\vdash \triangle A$, then by Lemma~\ref{lem:structural-elim} we have $\ell^\infty\ER^{str} \vdash \Ra \x : \triangle A$.
    As the only possible final rule for the proof of $\Ra \x : \triangle A$ in $\ell^\infty\ER^{str}$ is $\triagR$, we must have $\ell^\infty\ER^{str} \vdash \x S \y \Ra \y : A$, and thus a fortiori $\ell^\infty\ER \vdash \x S \y \Ra \y : A$.

    $(3) \Longrightarrow (1)$: Let $\ell^\infty\ER \vdash \x S \y \Ra \y : A$, then by $\triagR$ we obtain $\ell^\infty\ER \vdash \Ra \x : \triangle A$, then via Proposition~\ref{prop:NWFtoInd}, Theorem~\ref{thm:indSemanticSoundness} and Theorem~\ref{thm:SemanticCompl}(4) it follows that $\ER \vdash \triangle A$.
    From this it follows via Proposition~\ref{prop:type-ER-GLS} that $\GLS \vdash A$.
\end{proof}

\section{Conclusion} \label{sec:concl}

We introduced modular labelled sequent calculi for the bimodal provability logics $\CS$, $\CSM$, $\CSMP$ and $\ER$.
We showed that structural rules are eliminable in the calculi, and that they admit cyclic variants.
Completeness was shown via proof search which also yields a way of extracting countermodels when proof search fails. 
We proved soundness by finitising non-wellfounded proofs and showing that the resulting finite proofs are sound with respect to their respective relational semantics.

For future work, it will be interesting to advance this systematic study of bimodal provability logics further by considering additional bimodal provability logics, such as $\mathsf{CSMPR}$~\cite{Beklemishev94} or $\mathsf{EC}_\Sigma$~\cite{Beklemishev_1996}.
These would also lead to labelled proof-theoretic descriptions of other unimodal provability logics between $\GL$ and $\GLS$ by considering the $\triangle$-types of these logics, similar to Corollary~\ref{cor:ER-GLS}. 
While the labelled formalism is very flexible and can be useful for this approach, the obvious difficulty is that there is no known semantics for any other bimodal provability logic. 
Further, even if we had generalised semantics at hand, the full countermodel construction is rather tedious (see Definition~\ref{def:extended-ER-model}). 
However, this dependence of labels on the semantics can be overcome for both completeness and soundness.
Firstly, we can consider an explicit cut-elimination method for non-wellfounded proofs, instead of doing a countermodel construction, allowing for syntactic completeness~\cite{Acclavio_et_al24}.
Secondly, one can interpret labelled sequents as formulas (see e.g.~\cite{Strassburger_2013,MarinStrass_2014}), which is possible due to the natural restriction to tree-like sequents which can be equivalently formulated as nested sequents under certain restrictions~\cite{GoreRamanayake_2014}.
This enables one to show soundness with respect to the Hilbert axiomatisation directly.
We highlight here that we expect the rules for some frame-incomplete logics also to be rather simple, akin to the rules of $\ER$.
There are still further provability logics with multiple modalities of interest, such as $\GLP$~\cite{Japaridze85} or non-standard provability predicates~\cite{Henk16}, all of which are frame-incomplete, and for which we might find a rather easy proof-theoretic description using our approach. 
A general question regarding frame-incomplete logics is whether we can describe a more general class of frame-incomplete logics using our formalism.

\bibliography{references}

@InProceedings{Acclavio_et_al24,
  author =   {Acclavio, Matteo and Curzi, Gianluca and Guerrieri, Giulio},
  title =   {{Infinitary Cut-Elimination via Finite Approximations}},
  booktitle =    {32nd EACSL Annual Conference on Computer Science Logic (CSL 2024)},
  pages =   {8:1--8:19},
  series =   {Leibniz International Proceedings in Informatics (LIPIcs)},
  ISBN =    {978-3-95977-310-2},
  ISSN =    {1868-8969},
  year =    {2024},
  volume =  {288},
  editor =   {Murano, Aniello and Silva, Alexandra},
  publisher =   {Schloss Dagstuhl -- Leibniz-Zentrum für Informatik},
  address =  {Dagstuhl, Germany},
  doi =     {10.4230/LIPIcs.CSL.2024.8}
}

@article{Artemov80,
    author={S. N. Artemov},
    title={Arithmetically complete modal theories},
    journal={Semiotika i Informatika},
    volume={14},
    pages={115-133},
    place={VINITI, Moscow},
    year={1980},
    note={In Russian. English translation in: \textit{Amer. Math. Soc. Transl.} (2), 135: 39–54, 1987}
}

@incollection{Artemov2005,
title = {Provability Logic},
author = {Artemov, Sergei N. and Beklemishev, Lev D.},
editor = {Gabbay, D.M. and Guenthner, F.},
booktitle = {Handbook of Philosophical Logic, 2nd Edition},
publisher = {Springer Netherlands},
address = {Dordrecht},
pages = {189--360},
isbn = {978-1-4020-3521-0},
doi = {10.1007/1-4020-3521-7_3},
year = {2005}
}

@article{Beklemishev_1990,
doi = {10.1070/IM1990v035n02ABEH000701},
year = {1990},
month = {4},
publisher = {},
volume = {35},
number = {2},
pages = {247},
author = {L D Beklemishev},
title = {On the Classification of Propositional Provability Logics},
journal = {Mathematics of the USSR-Izvestiya}
}

@incollection{Beklemishev10,
title = {On Topological Models of {GLP}},
booktitle = {Ways of Proof Theory},
author = {Lev Beklemishev and Guram Bezhanishvili and Thomas Icard},
editor = {Ralf Schindler},
publisher = {De Gruyter},
address = {Berlin, Boston},
pages = {135--156},
doi = {10.1515/9783110324907.135},
isbn = {9783110324907},
year = {2010}
}

@article{Beklemishev94,
title = {On bimodal logics of provability},
journal = {Annals of Pure and Applied Logic},
volume = {68},
number = {2},
pages = {115-159},
year = {1994},
issn = {0168-0072},
doi = {https://doi.org/10.1016/0168-0072(94)90071-X},
author = {Lev D. Beklemishev}
}

@article{Beklemishev_1996, 
title={Bimodal logics for extensions of arithmetical theories}, 
volume={61}, 
ISSN={0022-4812, 1943-5886}, 
DOI={10.2307/2275600}, 
number={1}, 
journal={Journal of Symbolic Logic}, 
author={Beklemishev, Lev D.}, 
year={1996}, 
pages={91–124} 
}

@book{BlackburnRijkeVenema01, 
series={Cambridge Tracts in Theoretical Computer Science}, 
title={Modal Logic}, 
publisher={Cambridge University Press}, 
author={Blackburn, Patrick and Rijke, Maarten de and Venema, Yde}, 
year={2001}, 
collection={Cambridge Tracts in Theoretical Computer Science}
}

@book{Boolos94, 
place={Cambridge}, 
title={The Logic of Provability}, 
publisher={Cambridge University Press}, 
author={Boolos, George S.}, 
year={1994}
}

@article{Carlson86,
  title={Modal logics with several operators and provability interpretations},
  author={Timothy J. Carlson},
  journal={Israel Journal of Mathematics},
  year={1986},
  volume={54},
  pages={14-24},
  doi = {https://doi.org/10.1007/BF02764872}
}

@InProceedings{Das_et_al24,
  author =   {Das, Anupam and van der Giessen, Iris and Marin, Sonia},
  title =   {{Intuitionistic {G\"{o}del}-{L\"{o}b} Logic, \`{a} la {Simpson}: Labelled Systems and Birelational Semantics}},
  booktitle =   {32nd EACSL Annual Conference on Computer Science Logic (CSL 2024)},
  pages =   {22:1--22:18},
  series =   {Leibniz International Proceedings in Informatics (LIPIcs)},
  ISBN =    {978-3-95977-310-2},
  ISSN =    {1868-8969},
  year =    {2024},
  volume =  {288},
  editor =   {Murano, Aniello and Silva, Alexandra},
  publisher =   {Schloss Dagstuhl -- Leibniz-Zentrum für Informatik},
  address =  {Dagstuhl, Germany},
  doi =     {10.4230/LIPIcs.CSL.2024.22}
}

@incollection{Esakia_1981,
  author = {Leo Esakia},
  title = {Diagonal constructions, {L{\"o}b}'s formula and {Cantor}'s scattered spaces},
  booktitle = {Studies in logic and semantics},
  pages = {128--143},
  publisher = {Metsniereba},
  address = {Tbilisi},
  year = {1981},
  note = {In Russian}
}

@article{Feferman_1960, 
title={Arithmetization of metamathematics in a general setting}, 
volume={49}, 
ISSN={0016-2736}, 
DOI={10.4064/fm-49-1-35-92}, 
journal={Fundamenta Mathematicae}, 
publisher={Institute of Mathematics, Polish Academy of Sciences}, 
author={Feferman, S.}, 
year={1960}, 
pages={35–92} }

@book{Gabbay_1996,
  address={Oxford},
  series={Oxford logic guides},
  title={Labelled deductive systems},
  ISBN={978-0-19-853833-2},
  callNumber={511.3},
  publisher={Clarendon press},
  author={Gabbay, Dov M.},
  year={1996},
  collection={Oxford logic guides},
  doi ={10.1093/oso/9780198538332.001.0001}
}

@InProceedings{GargGenoveseNegri_2012,
    author = 	 {Deepak Garg and Valerio Genovese and Sara Negri},
    title = 	 {Countermodels from Sequent Calculi in Multi-Modal Logics},
    booktitle =  {Proceedings of the Twenty-Seventh Annual IEEE Symposium on Logic in Computer Science (LICS 2012)},
    year =	 {2012},
    month =	 {6}, 
    pages =      {315--324},
    location =   {Dubrovnik, Croatia}, 
    publisher =	 {IEEE Computer Society Press}
  }

@article{GoreRamanayake_2012,
  title={Valentini's Cut-Elimination for Provability Logic Resolved},
  volume={5},
  DOI={10.1017/S1755020311000323},
  number={2},
  journal={The Review of Symbolic Logic},
  author={Gor\'e, Rajeev and Ramanayake, Revantha},
  year={2012},
  pages={212–238}
}

@article{GoreRamanayake_2014,
    author = {Goré, Rajeev and Ramanayake, Revantha},
    year = {2014},
    month = {01},
    pages = {},
    title = {Labelled tree sequents, tree hypersequents and nested (Deep) sequents},
    volume = {9},
    journal = {Advances in Modal Logic}
}

@inproceedings{Goreetal_2021,
  author =   {Gor{\'e}, Rajeev and Ramanayake, Revantha and Shillito, Ian},
  editor =   {Das, Anupam and Negri, Sara},
  title =   {Cut-Elimination for Provability Logic by Terminating Proof-Search: Formalised and Deconstructed Using Coq},
  booktitle =   {Automated Reasoning with Analytic Tableaux and Related Methods},
  year =   {2021},
  publisher =   {Springer International Publishing},
  address =   {Cham},
  pages =   {299--313},
  isbn =   {978-3-030-86059-2},
  doi = {10.1007/978-3-030-86059-2_18}
}

@article{Henk16,
  title={Slow and Ordinary Provability for Peano Arithmetic},
  author={Paula Henk and Fedor Pakhomov},
  journal={arXiv: Logic},
  year={2016},
  doi = {10.48550/arXiv.1602.01822}
}

@inproceedings{Japaridze85,
  title={The polymodal logic of provability},
  author={Japaridze, Giorgi K.},
  booktitle={Intensional Logics and Logical Structure of Theories: Material from the fourth Soviet--Finnish Symposium on Logic, Telavi},
  pages={16--48},
  year={1985},
  note={In Russian.}
}

@book{Kanger57,
    address = {Stockholm,},
    author = {Kanger, Stig},
    editor = {},
    publisher = {Almqvist \& Wiksell},
    title = {Provability in Logic},
    year = {1957},
    doi = {10.1007/978-94-010-0500-5_2}
}

@article{Loeb_1955,
  title={Solution of a problem of Leon Henkin},
  author={L{\"o}b, Martin Hugo},
  journal={The Journal of Symbolic Logic},
  volume={20},
  number={2},
  pages={115--118},
  year={1955},
  publisher={Cambridge University Press},
  doi = {10.2307/2266895}
}

@inproceedings{MarinStrass_2014,
  title={Label-free Modular Systems for Classical and Intuitionistic Modal Logics},
  author={Sonia Marin and Lutz Stra{\ss}burger},
  booktitle={Advances in Modal Logic},
  year={2014},
  volume = {10},
  publisher={College Publications},
  editor = {Rajeev Goré, Barteld Kooi, Agi Kurucz},
  url = {http://www.aiml.net/volumes/volume10/Marin-Strassburger.pdf}
}

@misc{SierraMirandaStuder_2026,
  title={Proof Theory for Bimodal Provability Logics},
  author={Borja Sierra Miranda and Thomas Studer},
  year={2026},
  eprint={2605.12351},
  archivePrefix={arXiv},
  primaryClass={math.LO}
}

@article{Negri05,
  author = {Negri, Sara},
  year = {2005},
  pages = {507-544},
  title = {Proof Analysis in Modal Logic},
  volume = {34},
  journal = {Journal of Philosophical Logic},
  doi = {10.1007/s10992-005-2267-3}
}

@article{Niwinski96,
  title={Games for the $\mu$-calculus},
  journal={Theoretical Computer Science},
  volume={163},
  number={1},
  pages={99-116},
  year={1996},
  issn={0304-3975},
  doi={10.1016/0304-3975(95)00136-0},
  author={Damian Niwiński and Igor Walukiewicz}
}

@article{SambinValentini_1982,
  title={The modal logic of provability. The sequential approach},
  volume={11},
  ISSN={0022-3611, 1573-0433},
  DOI={10.1007/BF00293433},
  number={3},
  journal={Journal of Philosophical Logic},
  author={Sambin, Giovanni and Valentini, Silvio},
  year={1982},
  month=aug,
  pages={311–342}
}

@book{Segerberg_1971,
  title={An essay in classical modal logic},
  author={Segerberg, Karl Krister},
  year={1971},
  publisher={Stanford University}
}

@article{Shamkanov14,
  title={Circular proofs for the Gödel-Löb provability logic},
  volume={96},
  ISSN={1573-8876},
  DOI={10.1134/s0001434614090326},
  number={3–4},
  journal={Mathematical Notes},
  publisher={Pleiades Publishing Ltd},
  author={Shamkanov, D. S.},
  year={2014},
  month=sep,
  pages={575–585}
}

@PHDTHESIS{Simpson1994,
  title = {The Proof Theory and Semantics of Intuitionistic Modal Logic},
  author = {Simpson, Alex K.},
  year = {1994},
  school = {University of Edinburgh}
}

@book{Smorynski85,
  author = {Smoryński, Craig},
  year = {1985},
  title = {Self-Reference and Modal Logic},
  isbn = {9780387962092},
  publisher = {Springer-Verlag},
  place = {Berlin},
  doi = {10.1007/978-1-4613-8601-8}
}

@article{Solovay76,
  title={Provability interpretations of modal logic},
  volume={25},
  rights={http://www.springer.com/tdm},
  ISSN={0021-2172, 1565-8511},
  DOI={10.1007/BF02757006},
  number={3–4},
  journal={Israel Journal of Mathematics},
  author={Solovay, Robert M.},
  year={1976},
  month=sep,
  pages={287–304}
}

@inproceedings{Strassburger_2013,
  author = {Stra{\ss}burger, Lutz},
  booktitle = {Foundations of Software Science and Computation Structures},
  editor = {Pfenning, Frank},
  pages = {209--224},
  publisher = {Springer Berlin Heidelberg},
  title = {Cut Elimination in Nested Sequents for Intuitionistic Modal Logics},
  year = {2013},
  doi = {10.1007/978-3-642-37075-5_14}
}

@article{Takano_1992,
  title={Subformula property as a substitute for cut-elimination in modal propositional logics},
  author={Takano, Mitio},
  journal={Mathematica Japonica},
  volume={37},
  pages={1129--1145},
  year={1992}
}

@book{Vigano2000,
    address = {Boston},
    author = {Luca Vigan{\`o}},
    editor = {},
    publisher = {Kluwer Academic Publishers},
    title = {Labelled Non-Classical Logics},
    year = {2000},
    doi = {10.1007/978-1-4757-3208-5}
}

@phdthesis{Visser81,
  title={Aspects of diagonalization and provability},
  author={Visser, Albert},
  year={1981},
  school={University of Utrecht},
    place={Utrecht, Netherlands}
}

@article{Visser95,
  author = {Albert Visser},
  doi = {https://doi.org/10.1016/0168-0072(93)E0079-4},
  issn = {0168-0072},
  journal = {Annals of Pure and Applied Logic},
  note = {A Tribute to Dirk van Dalen},
  number = {1},
  pages = {109-142},
  title = {A course on bimodal provability logic},
  volume = {73},
  year = {1995}
}

@book{Lemmon1977,
    author = {Lemmon , E.J.},
    title = {An Introduction to Modal Logic},
    publisher = {Blackwell},
    place = {Oxford},
    year = {1977},
    note = {In collaboration with Dana Scott.}
}
\bibliographystyle{alphaurl}
\end{document}